\documentclass[letterpaper,11pt]{article}

\usepackage{color}
\usepackage{latexsym}

\usepackage{verbatim}

\usepackage{amsmath}
\usepackage{amssymb}

\usepackage{amsfonts}
\usepackage{graphicx}

\usepackage{ifpdf}
\ifpdf
\else
\fi

\newtheorem{theorem}{Theorem}[section]
\newtheorem{lemma}[theorem]{Lemma}

\newtheorem{corollary}[theorem]{Corollary}

\newenvironment{proof}{{\bf Proof:\ }}{\hfill$\Box$\medskip}

\newcommand{\ignore}[1]{}
\newcommand{\remove}[1]{}

\newcommand{\etal}{{et al.\ }}

\newcommand{\naive}{\mbox{\it na\"{i}ve \rm 1p}}
\newcommand{\lazy}{\mbox{\rm 2p}}
\newcommand{\eager}{\mbox{\rm 1p}}

\newcommand{\WHILE}{{\tt while}}

\newcommand{\FOR}{{\tt for}}

\newcommand{\IF}{{\tt if}}
\newcommand{\ELSE}{{\tt else}}
\newcommand{\RETURN}{{\tt return}}

\newcommand{\link}{\mbox{{\it link}}}

\newcommand{\meld}{\mbox{{\it meld}}}
\newcommand{\hinsert}{\mbox{{\it insert}}}
\newcommand{\hdelete}{\mbox{{\it delete}}}
\newcommand{\deletemin}{\mbox{{\it delete}-{\it min}}}
\newcommand{\makeheap}{\mbox{{\it make}-{\it heap}}}

\newcommand{\decreasekey}{\mbox{{\it decrease}-{\it key}}}

\newcommand{\addchild}{\mbox{{\it add}-{\it child}}}

\newcommand{\findmin}{\mbox{{\it find}-{\it min}}}

\newcommand{\makenode}{\mbox{{\it make}-{\it node}}}
\newcommand{\dorankedlinks}{\mbox{{\it do}-{\it ranked}-{\it links}}}
\newcommand{\dounrankedlinks}{\mbox{{\it do}-{\it unranked}-{\it links}}}
\newcommand{\newnode}{\mbox{{\it new}-{\it node}}}

\newcommand{\NULL}{\mbox{\it null}}

\newcommand{\rank}{\mbox{{\it rank}}}

\newcommand{\maxrank}{\mbox{{\it max\/}-{\it rank}}}
\newcommand{\node}{\mbox{{\it node}}}

\newcommand{\ITEM}{\mbox{{\it item}}}
\newcommand{\key}{\mbox{{\it key}}}

\newcommand{\next}{\mbox{{\it next}}}

\newcommand{\child}{\mbox{{\it child}}}

\newenvironment{mytabbing}
  {\setlength{\topsep}{0pt}%
   \setlength{\partopsep}{5pt}%
   \tabbing}
  {\endtabbing}

\begin{document}

\title{Hollow Heaps \thanks{A preliminary version of the paper appeared in ICALP 2015.}}

\author{
Thomas Dueholm Hansen\thanks{Department of Computer Science,
Aarhus University, Denmark. Supported by The Danish Council for
Independent Research $|$ Natural Sciences (grant no. 12-126512); and the Sino-Danish Center for the Theory of Interactive Computation, funded by the Danish National Research Foundation and the National Science Foundation of China (under the grant 61061130540). E-mail: {\tt tdh@cs.au.dk}. }
\and
Haim Kaplan\thanks{Blavatnik School of
Computer Science, Tel Aviv University,
  Israel. Research supported by
 the
Israel Science Foundation grants no.\ 822-10 and 1841/14,
 the
German-Israeli Foundation for Scientific Research and Development (GIF) grant no.\ 1161/2011, and
the Israeli Centers of Research Excellence (I-CORE) program (Center
No.\ 4/11). E-mail: {\tt haimk@post.tau.ac.il}.}
\and Robert E. Tarjan\thanks{Department of Computer Science, Princeton University,
Princeton, NJ 08540, USA and Intertrust Technologies, Sunnyvale, CA 94085, USA.}
\and Uri Zwick\thanks{Blavatnik School of
Computer Science, Tel Aviv University,
  Israel. Research supported by BSF grant no.\ 2012338 and by
The Israeli Centers of Research Excellence (I-CORE) program (Center
No.\ 4/11). E-mail: {\tt
    zwick@tau.ac.il}.}}


\maketitle

\begin{abstract}\noindent
We introduce the \emph{hollow heap}, a very simple data structure with the same amortized efficiency as the classical Fibonacci heap.  All heap operations except \hdelete\ and \deletemin\ take $O(1)$ time, worst case as well as amortized; \hdelete\ and \deletemin\ take $O(\log n)$ amortized time on a heap of $n$ items. Hollow heaps are by far the simplest structure to achieve this.  Hollow heaps combine two novel ideas: the use of lazy deletion and re-insertion to do \decreasekey\ operations, and the use of a dag (directed acyclic graph) instead of a tree or set of trees to represent a heap.  Lazy deletion produces hollow nodes (nodes without items), giving the data structure its name.
\end{abstract}



\section{Introduction}\label{sec:intro}

\pagestyle{plain}
\setcounter{page}{1}

A \emph{heap} is a data structure consisting of a set of \emph{items}, each with a \emph{key} selected from a totally ordered universe.  Heaps support the following operations:

\begin{description}
\item[$\makeheap( )$:] Return a new, empty heap.

\item[$\findmin(h):$] Return an item of minimum key in heap~$h$, or \NULL\ if~$h$ is empty.

\item[$\hinsert(e, k, h)$:]
Return a heap formed from heap~$h$ by inserting
item~$e$, with  key~$k$. Item~$e$ must be in no heap.

\item[$\deletemin(h)$:]
Return a heap formed from non-empty heap~$h$  by deleting
the item returned by $\findmin(h)$.

\item[$\meld(h_1, h_2)$:] Return a heap containing all items in item-disjoint heaps $h_1$ and~$h_2$.

\item[$\decreasekey(e, k, h)$:] Given that~$e$ is an item in heap~$h$ with key greater than~$k$,
return a heap formed from~$h$ by changing the key of~$e$ to~$k$.

\item[$\hdelete(e, h):$] 
Return a heap formed by deleting~$e$, assumed to be in~$h$, from~$h$.
\end{description}

The original heap~$h$ passed to $\hinsert$, $\deletemin$, $\decreasekey$, and $\hdelete$, and the heaps~$h_1$ and~$h_2$ passed to $\meld$, are destroyed by the operations.
Heaps do \emph{not} support search by key; operations \decreasekey\ and \hdelete\ are given the location of item~$e$ in heap~$h$. The parameter~$h$ can be omitted from \decreasekey\ and \hdelete, but then
to make \decreasekey\ operations efficient if there are intermixed \meld\ operations, a separate disjoint set data structure is needed to keep track of the partition of items into heaps.
(See the discussion in~\cite{KaShTa02b}.)

Fredman and Tarjan \cite{FrTa87} invented the \emph{Fibonacci heap}, an implementation of heaps that supports \deletemin\ and \hdelete\ on an $n$-item heap in $O(\log n)$ amortized time and each of the other operations in $O(1)$ amortized time.  Applications of Fibonacci heaps include a fast implementation of Dijkstra's shortest path algorithm \cite{Di59,FrTa87} and fast algorithms for undirected and directed minimum spanning trees \cite{Edmonds67,GaGaSpTa86}.  Since the invention of Fibonacci heaps, a number of other heap implementations with the same amortized time bounds have been proposed \cite{Brodal96,BrLaTa12,Chan13,Elmasry10,HaSeTa11,Hoyer95,KaTa08,Peterson87,Takaoka03}.  Notably, Brodal \cite{Brodal96} invented a very complicated heap implementation that achieves the time bounds of Fibonacci heaps in the worst case.  Brodal \etal \cite{BrLaTa12} later simplified this data structure, but it is still significantly more complicated than any of the amortized-efficient structures.
For further discussion of these and related results, see \cite{HaSeTa11}. We focus here on the \emph{amortized} efficiency of heaps.

In spite of its many competitors, Fibonacci heaps remain one of the simplest heap implementations to describe and code, and are taught in numerous undergraduate and graduate data structures courses.
We present {\em hollow heaps}, a data structure that we believe surpasses Fibonacci heaps in its simplicity.
Our {final} data structure has two novelties: it uses lazy deletion to do \decreasekey\ operations in a simple and natural way, avoiding the \emph{cascading cut} process used by Fibonacci heaps, and it represents a heap by a dag (directed acyclic graph) instead of a tree or a set of trees. The amortized analysis of hollow heaps is simple, yet non-trivial. We believe that simplifying fundamental data structures, while retaining their performance, is an important endeavor.

In a Fibonacci heap, a \decreasekey\ produces a heap-order violation if the new key is less than that of the parent node.  This causes a \emph{cut} of the violating node and its subtree from its parent.  Such cuts can eventually destroy the ``balance" of the data structure.  To maintain balance, each such cut may trigger a cascade of cuts at ancestors of the originally cut node.  The cutting process results in loss of information about the outcomes of previous comparisons. It also makes the worst-case time of a \decreasekey\ operation $\Theta(n)$ (although  modifying the data structure reduces this to $\Theta(\log n)$; see e.g., \cite{KTZ14}).  In a hollow heap, the item whose key decreases is merely moved to a new node, preserving the existing structure.  Doing such lazy deletions carefully is what makes hollow heaps simple but efficient.


Starting from the ideas used in Fibonacci heaps, we develop hollow heaps in three steps.  First, we show how \emph{hollow nodes}, with an appropriate way of moving children and setting ranks, can replace the cascading cut process that makes Fibonacci heaps efficient.  Second, we replace the set of trees representing a heap by a single tree, by allowing \emph{unranked links} (links between roots of different ranks).  This idea was used before in \cite{HaSeTa11,KTZ14} and is orthogonal to the idea of using hollow nodes.  Third, we obtain our final data structure by showing how to avoid moving children of hollow nodes at all.  To do this, we represent a heap by a (tree-like) dag rather than a tree: each node can have up to two parents, rather than just one.

The remainder of our paper consists of 6 sections.
Section~\ref{sec:multi-root}  presents a multi-root version of hollow heaps, the first of the three steps mentioned above.  Section~\ref{sec:one-root}  presents a one-root version of hollow heaps.  Section~\ref{sec:lazy}  presents the final version of our data structure, which replaces the movement of children by the use of a dag representation.
Section~\ref{sec:rebuilding} describes a rebuilding process that can be used to improve the time and space efficiency of hollow heaps.
Section~\ref{sec:lazy-implementation}  gives implementation details for the data structure in Section~\ref{sec:lazy}.
Section~\ref{sec:good-bad} explores the design space of hollow heaps, identifying variants that are efficient and variants that are not.
%

\section{Multi-root hollow heaps}\label{sec:multi-root}

As in the classic Fibonacci heap data structure, we represent each heap by a set of heap-ordered trees.  In a Fibonacci heap, each node stores exactly one item, and each item is in exactly one node.  We relax this invariant to allow nodes that do not hold items.  We call a node \emph{full} if it holds an item and \emph{hollow} if not.  A newly created node is initially full, but can later become hollow by having its item moved to a new node or deleted.  A hollow node stays hollow until it is destroyed.  Since items are moved among nodes, the structure is necessarily \emph{exogenous} rather than \emph{endogenous} \cite{Tarjan83}: nodes \emph{hold} items rather than \emph{being} items.

If~$u$ is a node, we denote by $u.item$ the item held by~$u$ if~$u$ is full, or $null$ if~$u$ is hollow.  If~$e$ is an item, we denote by $e.node$ the node holding~$e$.  Each node~$u$ has a key $u.key$ associated with it: \footnote{Alternatively, one can store keys with items rather than with nodes.  Which alternative is preferable is primarily an experimental question.} if~$u$ is full, $u.key$ is the current key of $u.item$; if~$u$ is hollow, $u.key$ is the key of the item once held by~$u$, just before it was moved from~$u$.


We organize nodes into rooted trees.  If $(v, w)$ is a tree arc, we say that~$v$ is the \emph{parent} of~$w$ and~$w$ is a \emph{child} of~$v$.  (We direct tree arcs from parent to child.)  If there is a path of zero or more arcs from~$v$ to~$w$ in a tree, we say that~$v$ is an \emph{ancestor} of~$w$ and~$w$ is a \emph{descendant} of~$v$.  A node with no parent is a \emph{root}; a node with no children is a \emph{leaf}.  A tree is \emph{heap-ordered} if and only if for every arc $(v, w)$, $v.key \le w.key$ (whether or not~$v$ and~$w$ are hollow). Heap-order implies that the root has a minimum key.

We maintain the invariant that the set of trees representing a heap is either empty or contains a full root whose key is minimum among those of all nodes in the trees, full and hollow.  We maintain a pointer to such a root.  We call the node indicated by this pointer the \emph{minimum node} of the heap.

A generic way of implementing the heap update operations is via the \link\ primitive.  Given two full roots~$u$ and~$v$, $\link(u, v)$ compares their keys and makes the root of larger key a child of the other, breaking a tie arbitrarily. 
The new child is the \emph{loser} of the link; its new parent is the \emph{winner}.  Linking eliminates one full root, preserves heap order, and gives the loser a parent.  As in Fibonacci heaps, we do links only during \deletemin\ or \hdelete\ operations that delete the item in the minimum node.

To do \makeheap, return an empty forest.  To do \findmin, return the item in the minimum node.  To \meld\ two heaps, if one is empty return the other; otherwise, unite their sets of trees and update the minimum node.
To \hinsert\ an item into a heap, create a new node, store the item in it (making the node full), and \meld\ the resulting one-node heap with the existing heap.

The \decreasekey\ operation uses a lazy form of deletion.  To decrease the key of item~$e$ in heap $h$ to $k$, let $u = e.node$ and do the following: create a new node~$v$; move~$e$ from~$u$ to~$v$ (making~$v$ full and~$u$ hollow); set $v.key = k$; move some or all of the children of~$u$, and their subtrees, to~$v$; \meld\ the one-root heap consisting of~$v$ and its descendants with the existing heap.  The choice of which children to move is a critical design decision that we address shortly.  As an optimization, if~$u$ is a root one can avoid creating a new node by merely setting $u.key = k$ and updating the minimum node.

To do \deletemin, delete the item in the minimum node.
To delete an item~$e$, remove~$e$ from the node $u = e.node$ holding it, making~$u$ hollow.
 This completes the deletion unless~$u$ is the minimum node.
After deleting the item in the minimum node during either a \deletemin\ or \hdelete,
 proceed as follows: while there is a hollow root, destroy such a root, making each of its children into a root (unite the set of such new roots with the set of old roots).  Once all roots are full, do zero or more links to reduce the number of roots.
Make each remaining tree into a heap, and meld these heaps in any order.  The choice of which links to do is another critical design decision that we address shortly.


{\bf Remark:}
This implementation of \hdelete\ allows hollow roots. One can keep all roots full by proceeding as in the case of deletion of the item in the
minimum node whenever an item in \emph{any} root is deleted.

As in Fibonacci heaps, we make the data structure efficient by using \emph{node ranks}.
\footnote{Fredman \cite{Fredman99} has shown that obtaining a constant amortized bound for \decreasekey\ in a data structure like ours
\emph{requires} storing $\Omega(n\log\log n)$ extra bits of information, such as node ranks, although his result has significant technical restrictions. See also
Iacono and {\"{O}}zkan \cite{IaconoO14}.}
We give each node~$u$ a non-negative integer rank $u.rank$.  When a node is created by an insertion, its initial rank is $0$.  We do links only between roots of equal rank.  Such a link increases the rank of the winner by one.  We call such a link a \emph{ranked link}.  A \hdelete\ that deletes the item in the minimum node does ranked links until none are possible (all roots have different ranks).

The remaining design decision is the choice of which children (and their subtrees) to move during a \decreasekey\ operation (and what rank to give the newly created node).  We achieve efficiency by maintaining the following \emph{rank invariant}:

\begin{quote}
A node~$u$ of rank $r$ has exactly $r$ children, of ranks $0, 1,\ldots, r - 1$, \emph{unless} $r > 2$ and~$u$ was made hollow by a \decreasekey, in which case~$u$ has exactly two children, of ranks $r - 2$ and $r - 1$.
\end{quote}

To maintain the rank invariant, when doing a \decreasekey\ operation that moves an item from node~$u$ to node~$v$, initialize $v.rank = \max\{0, u.rank - 2\}$, and move to~$v$ each child of~$u$ of rank less than $v.rank$ (along with their subtrees).  If $u.rank \ge 2$, $u$ retains its children of ranks $u.rank - 2$ and $u.rank - 1$, along with their subtrees; if $u.rank = 1$, $u$ retains its only child (of rank $0$).

The resulting data structure is the \emph{multi-root hollow heap}.
Figure~\ref{fig:multi-root} illustrates operations on a multi-root hollow heap.

\begin{figure}[t]
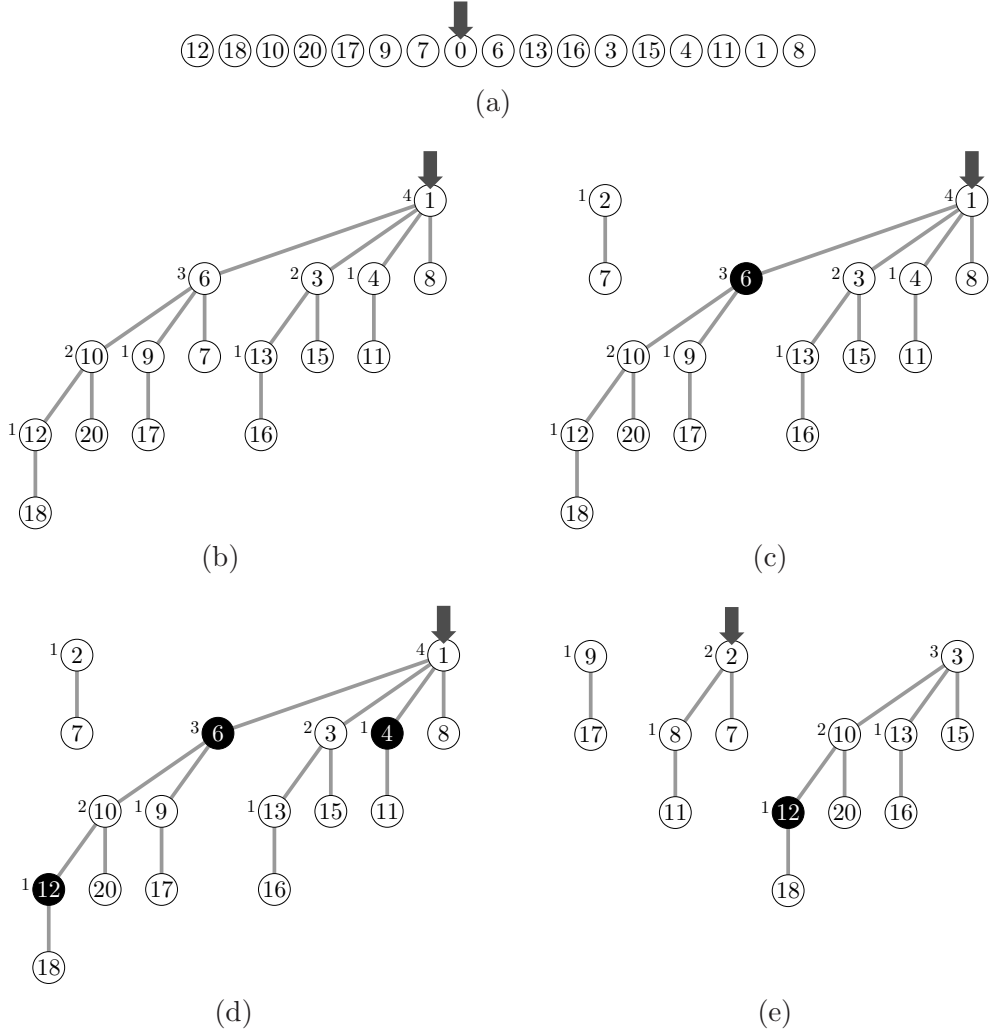

\begin{center}
\renewcommand{\tabcolsep}{0.2cm}
\begin{tabular}{c}
\includegraphics[scale=0.5]{figures2.17}\vspace{0.1cm} \\
(a) \vspace{0.35cm}
\end{tabular}
\begin{tabular}{cc}
\includegraphics[scale=0.5]{figures2.18} \hspace*{10pt}&\hspace*{10pt}
\includegraphics[scale=0.5]{figures2.19} \\
(b) \hspace*{15pt} & \hspace*{15pt} (c) \vspace{0.35cm}
\end{tabular}
\begin{tabular}{cc}
\includegraphics[scale=0.5]{figures2.20}  \hspace*{10pt}& \hspace*{10pt}
\includegraphics[scale=0.5]{figures2.21}  \\
(d)  \hspace*{15pt}&\hspace*{15pt} (e)
\end{tabular}
\end{center}
\caption{
Operations on a multi-root hollow heap.  Numbers in nodes are keys; black nodes are hollow. Numbers next to nodes are non-zero ranks. (a) The heap after the successive insertion of items with keys, 12, 18, 20, 20,\ldots. Each item forms a singleton tree. The arrow pointing to the item with key 0 is the minimum pointer.
(b) The heap after the deletion of the the minimum item, i.e., the item with key 0. (c) The heap after decreasing key 6 to 2. (d) The heap after deleting the items with keys 4 and 12. (e) The heap after a \deletemin\ operation that deletes the item with key~1.
}\label{fig:multi-root}
\end{figure}

\begin{theorem}\label{thm:multi-correct}
Multi-root hollow heaps correctly implement all the heap operations, maintain heap-ordered trees, and maintain the rank invariant.
\end{theorem}

\begin{proof}
The proof is straightforward by induction on the number of heap operations.
\end{proof}

The effect of our design decisions is to keep the trees \emph{balanced} in an appropriate sense, as we show by an argument like that for Fibonacci heaps but simpler and more direct.
Recall the definition of the Fibonacci numbers: $F_0 = 0$, $F_1 = 1$, $F_i = F_{i - 1} + F_{i - 2}$ for $i \ge 2$.  These numbers satisfy $F_{i + 2} \ge \phi^i$, where $\phi = (1 + \sqrt{5})/2$ is the \emph{golden ratio} \cite{Knuth98}.

\begin{lemma}\label{lem:multi-fib}
A node of rank~$r$ has at least $F_{r + 3} - 1$ descendants, both full and hollow.
\end{lemma}

\begin{proof}
The proof is by induction on~$r$. The claim is immediate for $r = 0$ and $r = 1$. If $r \ge 2$, the descendants of a node~$u$ of rank~$r$ include itself and all descendants of its children of ranks $r - 1$ and $r - 2$, which it has by the rank invariant, whether it is full or hollow.  By the induction hypothesis, $u$ has at least $1 + (F_{r + 2} - 1) + (F_{r + 1} - 1) = F_{r + 3} - 1$ descendants.
\end{proof}

\begin{corollary}\label{cor:multi-rank}
The rank of a node in a multi-root hollow heap of~$N$ nodes is at most $\log_\phi N$.
\end{corollary}

\begin{proof}
The corollary is immediate from Lemma~\ref{lem:multi-fib} since $F_{r + 3} - 1 \ge F_{r + 2} \ge \phi^r$ for $r \ge 0$.
\end{proof}

To obtain an efficient implementation of multi-root hollow heaps, we store the children of each node in a singly-linked list in decreasing order by rank.  We store the roots of each heap in a singly-linked circular list accessed via the minimum node.  Circular linking allows lists of roots to be catenated in constant time, allowing \meld\ to be done in constant time.  When doing a ranked link, we make the new child the first child of its new parent; this keeps the children in decreasing rank order.  When doing a \decreasekey\ that moves an item from a node~$u$ to a node~$v$, if $u.rank > 2$ we move to~$v$ all but the first two children of~$u$; if $u.rank \le 2$, we move no children.  Each heap operation except \hdelete\ takes constant time worst-case.  To find roots to link during a delete, we use an array of roots indexed by rank, as in Fibonacci heaps \cite{FrTa87}. (A single global array may be used for all heaps.)  The time for a \hdelete\ is $O(1)$ unless it deletes the item in the minimum node, in which case it takes
$O(H + T)$ time, where $H$ is the number of hollow roots destroyed and $T$ is the number of trees remaining after destruction of hollow roots but before any links.  After the links there is at most one tree per rank, totaling at most $\log_\phi N$ by Corollary~\ref{cor:multi-rank}.  Thus the number of links is at least $T - \log_\phi N$, and the number of melds is at most $\log_\phi N$. It follows that the time for the \hdelete\ is $O(1)$ per hollow root destroyed plus $O(1)$ per \link\ plus $O(\log N)$.


We bound the amortized time per operation by modifying the analysis of Fibonacci heaps.

\begin{theorem}\label{thm:multi-analysis}
The amortized time per multi-root hollow heap operation is $O(1)$ for each operation other than a \hdelete\ or \deletemin, and $O(\log N)$ per \hdelete\ or \deletemin\ on a heap of $N$ nodes.
\end{theorem}

\begin{proof}
The worst-case time per operation is $O(1)$ except for a \hdelete\ that deletes the item in the minimum node, which takes time $O(1)$ per hollow node destroyed plus $O(1)$ per \link\ plus $O(\log N)$.
We charge node destructions against the corresponding node creations, one per \hinsert\ and one per \decreasekey.  To obtain the theorem, it remains to bound the number of links.


To count links use a \emph{potential} function \cite{Tarjan85}.  We give each full node that is not a child of another full node a potential of one, and all other nodes (hollow nodes and full children of full nodes) a potential of zero.  The potential of a heap is the sum of the potentials of its nodes.  We define the \emph{amortized cost} of an operation to be the number of links plus the increase in potential.  Since the initial potential is zero (the data structure is empty) and the potential is always non-negative, the total number of links is at most the sum of the amortized costs of all the operations.


All the links are in \hdelete\ operations that delete the item in the minimum node, so the amortized cost of every other operation is the potential change.  This is one per \hinsert\ (one new root), zero per \findmin\ or \meld, and at most three per \decreasekey\ (one new root and at most two full children of a new hollow node).  In a \hdelete, each \link\ has an amortized cost of zero: it converts a full root into a full child of a full root, reducing the potential of the new child by one.  The only other potential change in a \hdelete\ is caused by removing an item.  This creates a hollow node~$u$, increasing the potential by one per full child of~$u$, at most $\log_\phi N$ in total.  We conclude that the number of links is $O(1)$ per heap operation plus $\log_\phi N$ per \hdelete\ on a heap of $N$ nodes.
\end{proof}

{\bf Remark:} The proof of Theorem~\ref{thm:multi-analysis} also gives a bound on the amortized number of comparisons per operation: at most one per \hinsert\ and \meld, three per \decreasekey, and $2\log_\phi N$ per \hdelete\ on a heap of $N$ nodes.

We have obtained multi-root hollow heaps from Fibonacci heaps by making a very simple change: when doing a \decreasekey\ operation at a node~$u$, instead of moving the entire subtree rooted at~$u$, we leave~$u$ and up to two of its children (and their subtrees) in place.  This preserves enough of the tree balance to avoid doing cascading cuts as in Fibonacci heaps, or cascading rank changes as in \cite{HaSeTa11,KTZ14}, or restructuring to reduce heap-order violations as in \cite{BrLaTa12,DrGaShTa88,KaTa08}.  As a result, \decreasekey\ operations take $O(1)$ time worst-case as well as amortized, and no additional state information, such as mark bits, is needed.
The Fibonacci numbers enter the analysis (Lemma~\ref{lem:multi-fib}) more directly, and the amortized analysis (Theorem~\ref{thm:multi-analysis}) becomes simpler.  The implementation is also simpler: singly-linked lists replace the doubly-linked lists needed in Fibonacci heaps, and no parent pointers are needed.

The price we pay for these changes is threefold: the data structure becomes exogenous rather than endogenous (items point to nodes and vice versa, rather than items being nodes); the amortized time of \hdelete\ and \deletemin\ becomes $O(\log N)$ rather than $O(\log n)$, and the space used by the data structure is $O(N)$ rather than $O(n)$, where $n$ is the number of items in the heap and $N$ is the number of nodes.  The latter is at most the number of \hinsert\ and \decreasekey\ operations.  We can reduce the time and space bounds to $O(\log n)$ and $O(n)$ by periodically rebuilding the data structure, as we discuss in Section~\ref{sec:rebuilding}.


Multi-root hollow heaps are asymptotically optimal, but they do not use all available key order information.  In particular, they do not keep track of all key comparisons done when updating minimum nodes during melds, nor do they keep track of the decreasing key order of nodes successively holding the same item.  In the next two sections we modify the data structure to use this additional information.

\section{One-root hollow heaps}\label{sec:one-root}

By allowing links between roots of different ranks, we can obtain a one-root version of hollow heaps.
This idea is quite general and is orthogonal to the idea of hollow nodes: it was used in \cite{HaSeTa11} to obtain a one-root version of rank-pairing heaps and in \cite{KTZ14} to obtain a one-root version of Fibonacci heaps.  Here we apply it to hollow heaps.
In the next section, we modify hollow heaps more drastically, to avoid moving children during \decreasekey\ operations.

As described in Section~\ref{sec:multi-root}, a \emph{ranked link} applies to two full roots of equal rank.  It makes the node of larger key a child of the node of smaller key (the winner) and increases the rank of the winner by one, breaking a tie arbitrarily.  In contrast, an \emph{unranked link} applies to any two full roots, whether or not their ranks are equal.  It makes the node of larger key a child of the node of smaller key and changes no ranks, breaking a tie arbitrarily.  In either kind of link, the new child is the loser and its parent is the winner.  The loser of a link is a \emph{ranked} or \emph{unranked} child if the link is ranked or unranked, respectively.  A child retains its ranked or unranked state when moved to a new parent by a \decreasekey\ operation; it can only change state by losing another link, which can only happen after its parent becomes hollow and is later destroyed.


A \emph{one-root hollow heap} is either empty or is a single heap-ordered tree whose root is full.  The root is the minimum node.  We do the heap operations as in multi-root hollow heaps, with the following changes.  To \meld\ two non-empty heaps, do an unranked link of their roots.  In a \decreasekey\ that moves an item from a node~$u$ to a new node~$v$, initialize $v.rank = \max\{0, u.rank - 2\}$, and move to~$v$ all \emph{ranked} children of~$u$ of ranks less than $v.rank$ (along with their subtrees), and any or all of the unranked children of~$u$ (along with their subtrees).


The rank invariant for one-root hollow heaps is the following:


\begin{quote}
A node~$u$ of rank $r$ has exactly $r$ ranked children, of ranks $0, 1,\ldots, r - 1$, unless $r > 2$ and~$u$ was made hollow by a \decreasekey, in which case~$u$ has exactly two ranked children, of ranks $r - 2$ and $r - 1$.  (A node can have any number of unranked children.)
\end{quote}

Figure~\ref{fig:one-root} illustrates operations on a one-root hollow heap.

\begin{figure}[t]
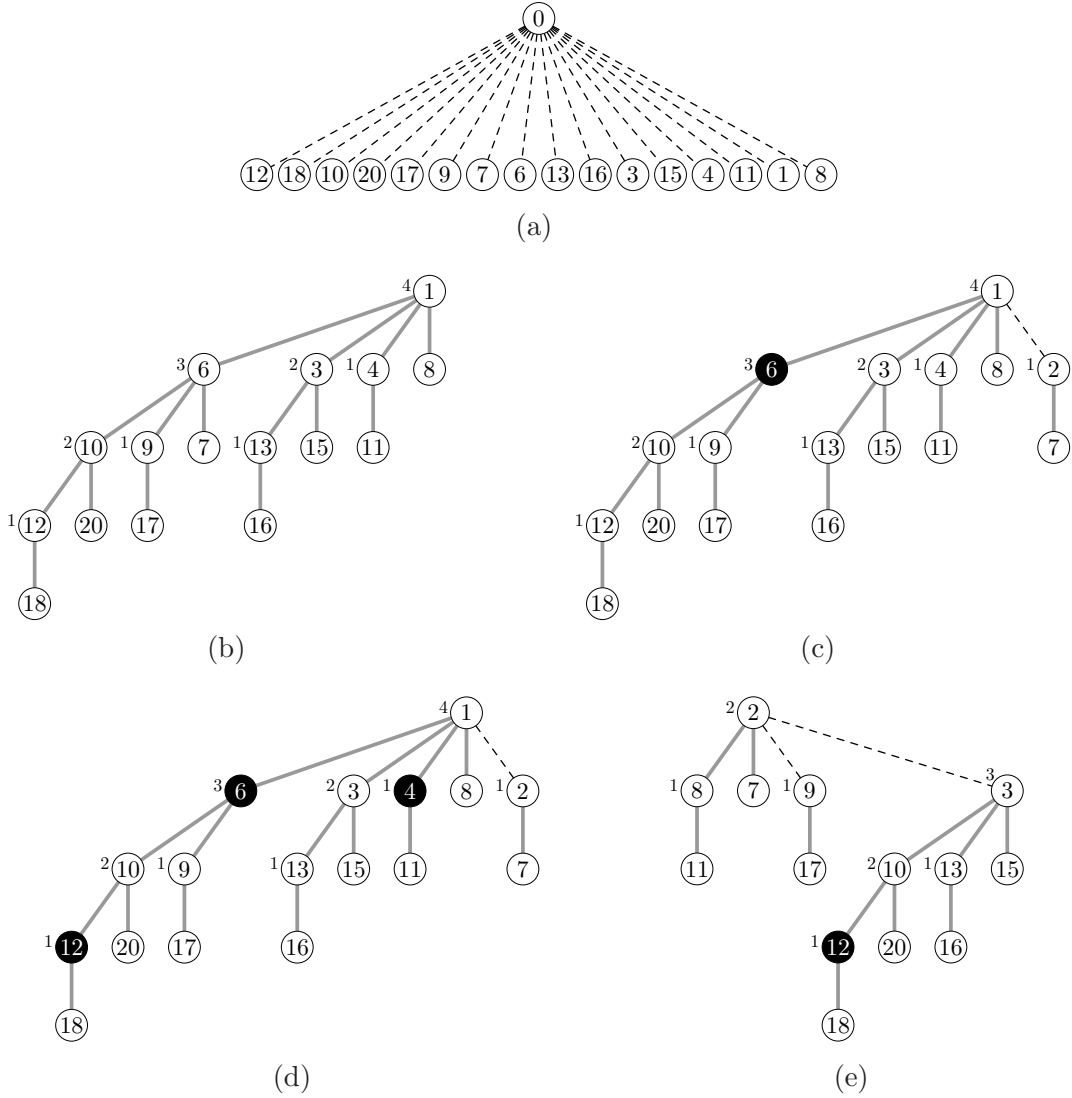

\begin{center}
\renewcommand{\tabcolsep}{0.2cm}
\begin{tabular}{c}
\includegraphics[scale=0.5]{figures2.22}\vspace{0.1cm}\\
(a) \vspace{0.35cm}
\end{tabular}
\begin{tabular}{cc}
\includegraphics[scale=0.5]{figures2.23} \hspace*{15pt}&\hspace*{15pt}
\includegraphics[scale=0.5]{figures2.24} \\
(b) \hspace*{15pt} & \hspace*{15pt} (c) \vspace{0.35cm}
\end{tabular}
\begin{tabular}{cc}
\includegraphics[scale=0.5]{figures2.25}  \hspace*{15pt}& \hspace*{15pt}
\includegraphics[scale=0.5]{figures2.26} \vspace{0.1cm} \\
(d)  \hspace*{15pt}&\hspace*{15pt} (e)
\end{tabular}
\end{center}
\caption{
Operations on a one-root hollow heap.  Numbers in nodes are keys; black nodes are hollow. Numbers next to nodes are non-zero ranks. Bold gray and dashed lines denote ranked and unranked links, respectively.
(a) The heap after the successive insertion of items with keys~$0,12,18,\ldots$. All links performed are unranked. (a) The heap after the deletion of the minimum item. All links performed are ranked. (c) The heap after decreasing key 6 to 2. (d) The heap after deleting the items with keys 4 and 12. (e) The heap after another \deletemin\ operation.
}\label{fig:one-root}
\end{figure}

\begin{theorem}\label{thm:eager-correct}
One-root hollow heaps correctly implement all the heap operations, maintain heap-ordered trees, and maintain the rank invariant.
\end{theorem}

\begin{proof}
The proof is straightforward by induction on the number of heap operations.
\end{proof}

\begin{lemma}\label{lem:eager-fib}
A node of rank $r$ in a one-root hollow heap has at least $F_{r + 3} - 1$ descendants, both full and hollow.
\end{lemma}

\begin{proof}
The proof is the same as that of Lemma~\ref{lem:multi-fib}.
\end{proof}

\begin{corollary}\label{cor:eager-rank}
The rank of a node in a one-root hollow heap of $N$ nodes is at most $\log_\phi N$.
\end{corollary}

\begin{proof}
The proof is the same as that of Corollary~\ref{cor:multi-rank}, using Lemma~\ref{lem:eager-fib}.
\end{proof}

To obtain an efficient implementation of one-root hollow heaps, we store the children of each node in a singly-linked circular list, accessed via the \emph{last} child on the list.  The ranked children are first, in decreasing order by rank, followed by the unranked children.  Making each list of children circular and accessing such a list via its last node allows doing links in constant time while maintaining the desired order of children: when doing a ranked link, add the new child to the \emph{front} of the list of its new parent, when doing an unranked link, add it to the \emph{back}.  We do \emph{not} need to store the states of children (ranked or unranked), only the node ranks.  When doing a \decreasekey\ that moves an item from a node~$u$ to a new node~$v$, if $u.rank > 2$ we move all but the first two children of~$u$ to~$v$, along with their subtrees.  This leaves with~$v$ the two ranked children of highest rank, and moves all the unranked children.  (The implementation is free to move none, some, or all of the unranked children.)

\begin{theorem}\label{thm:eager-analysis}
The amortized time per one-root hollow heap operation is $O(1)$ for each operation other than a \hdelete\ or \deletemin, and $O(\log N)$ per \hdelete\ or \deletemin\ on a heap of $N$ nodes.
\end{theorem}

\begin{proof}
The proof is like the proof of Theorem~\ref{thm:multi-analysis} but with a slightly different potential function: each full node has one unit of potential unless it is a full ranked child of a full node, in which case it has zero potential.  A ranked link reduces the potential by one, giving it an amortized cost of zero.  An unranked link does not change the potential giving it an amortized cost of one.  The amortized cost of an operation other than a delete or delete-min is $O(1)$ by the same argument as in the proof of Theorem~\ref{thm:multi-analysis}.  In a \hdelete\ or \deletemin, making a node hollow by deleting its item increases the potential by $O(\log N)$, since only ranked children increase in potential.  If the deletion is of the item in the minimum node, the ranked links have amortized cost zero, and there are $O(\log N)$ unranked links, at most one per rank.  It follows that the amortized cost of a \hdelete\ or \deletemin\ is $O(\log N)$.

\end{proof}

We conclude this section by discussing two options in doing links.  First, at the end of a delete that removes the item in the minimum node, we are free to do the melds in any order.  Repeatedly melding two trees of minimum rank is natural and easy to implement, but we do not have theory to justify this order.

Second, we can in certain cases increase the rank of the winner in an unranked link.  In particular, in an unranked link of two roots of equal rank, we can increase the rank of the winner by one, making the link into a ranked link; in an unranked link of two roots of unequal rank, if the root of smaller rank is the winner, we can increase its rank to any value no greater than the rank of the loser.  To allow increases of the second type, we need to change the algorithm slightly.  We define any link that increases the rank of the winner to be a ranked link, and any other link to be an unranked link.  As in the original implementation, the loser of a ranked or unranked link becomes the first or last child of its new parent, respectively.  When doing a \decreasekey\ that moves an item from a node~$u$ to a new node~$v$, if~$u$ has at most two children we merely set the rank of~$v$ to zero; if~$u$ has more than two children we leave the first two of them with~$u$, move the rest to~$v$ (preserving their order), and set the rank of~$v$ equal to that of its new first child.  This method maintains the following invariant:

\begin{quote}
A node of rank $r > 0$ either has a first child of rank at least~$r$, or it has a first child of rank $r-1$ and a second child of rank at least
$r-2$.
\end{quote}

An extension of our analysis shows that this method is correct and efficient.  With this method it is possible for a node to accumulate many more ranked children than its rank.  To handle this in the analysis, we give each full node additional potential equal to the maximum of zero and its number of ranked children minus its rank.  We omit the details of the analysis, since they are so similar to what we have already presented.

Whether either or both of these options improves performance is an experimental question; they certainly complicate the algorithm.

\section{Two-parent hollow heaps}\label{sec:lazy}

Our second and more drastic change to hollow heaps eliminates the movement of children during \decreasekey\ operations.  This change comes at the cost of converting the data structure from a tree or set of trees to a (treelike) \emph{dag} (directed acyclic graph).  During a \decreasekey\ that moves an item from a node~$u$ to a new node~$v$ (still of initial rank $\max\{0, u.rank - 2\}$), instead of moving certain children of~$u$ to~$v$, we merely make~$u$ a child of~$v$.  This gives~$u$ a \emph{second} parent.  We make this change to the one-root hollow heaps of Section~\ref{sec:one-root}, but it applies equally to the multi-root hollow heaps of Section~\ref{sec:multi-root}.  We call the resulting data structure the \emph{two-parent hollow heap}, in contrast to the data structures of Sections~\ref{sec:multi-root} and~\ref{sec:one-root}, which  we jointly call \emph{one-parent hollow heaps}.

To present the details of this idea, we extend our tree terminology to dags.  If $(v, w)$ is a dag arc, we say that~$v$ is the \emph{parent} of~$w$ and~$w$ is a \emph{child} of~$v$.  A node with no parent is a \emph{root}; a node with no children is a \emph{leaf}.  A dag whose nodes have keys is \emph{heap-ordered}
if and only if for every arc $(v, w)$, $v.key \le w.key$. That is: any topological order of the dag arranges the nodes in non-decreasing order by key.

A \emph{two-parent hollow heap} is either empty or is a heap-ordered dag with one full root and no hollow roots.  The root is the minimum node.  We do the heap operations as in one-root hollow heaps, with the following change: in a \decreasekey\ that moves an item from a node~$u$ to a new node~$v$, make~$u$ a child of~$v$; do not move any of the children of~$u$.
Figure~\ref{fig:hollow-heap} illustrates operations on a two-parent hollow heap.

\begin{figure}[t]
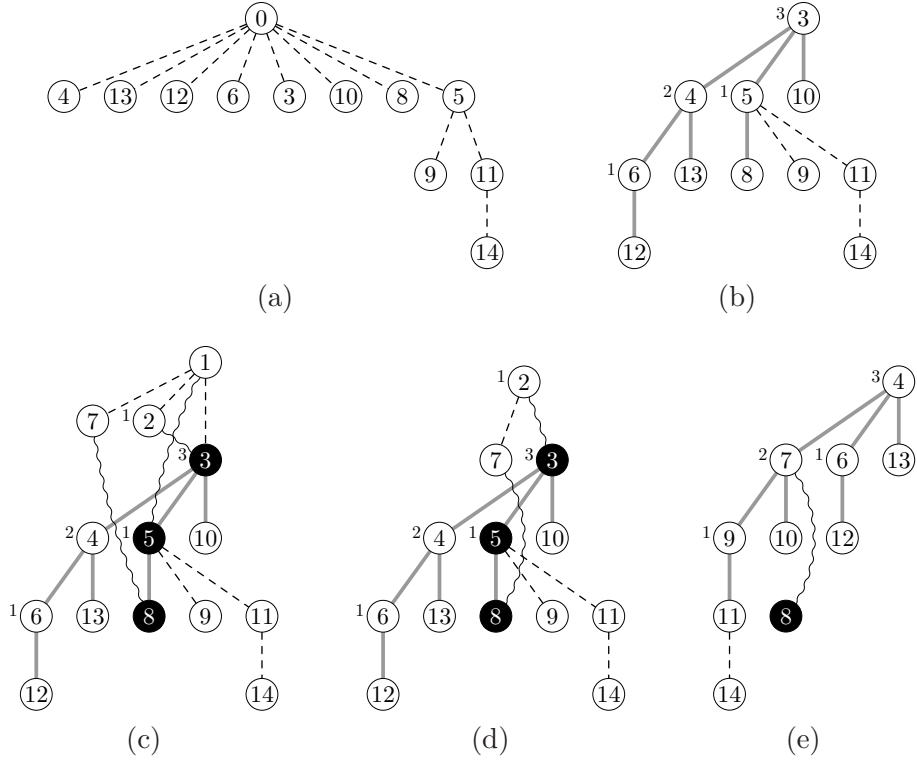

\begin{center}
\renewcommand{\tabcolsep}{0.2cm}
\begin{tabular}{cc}
\includegraphics[scale=0.5]{figures2.8} \hspace*{10pt}&\hspace*{10pt}
\includegraphics[scale=0.5]{figures2.9} \\
(a) \hspace*{10pt} & \hspace*{10pt} (b) \vspace{0.35cm}
\end{tabular}
\begin{tabular}{ccc}
\includegraphics[scale=0.5]{figures2.16}  \hspace*{5pt}& \hspace*{5pt}
\includegraphics[scale=0.5]{figures2.12}  \hspace*{5pt}& \hspace*{5pt}
\includegraphics[scale=0.5]{figures2.13}\\
(c)  \hspace*{5pt}&\hspace*{5pt} (d) \hspace*{5pt}&  \hspace*{5pt}(e)
\end{tabular}
\end{center}
\caption{
Operations on a two-parent hollow heap.  Numbers in nodes are keys; black nodes are hollow.  Bold gray, dashed, and squiggly lines denote ranked links, unranked links, and second parents, respectively.  Numbers next to nodes are non-zero ranks. (a) Successive insertions of items with keys 14, 11, 5, 9, 0, 8, 10, 3, 6, 12, 13, 4 into an initially empty heap. (b) After a \deletemin\ operation.  All links during the \deletemin\ are ranked.  (c) After decreases of key 5 to 1, then key 3 to 2, and then key 8 to 7.  (d) After a second \deletemin.  The only hollow node that becomes a root is the original root.  One unranked link, between the nodes holding keys 2 and 7 occurs.  (e) After a third \deletemin.  The hollow nodes with keys 3 and 5 become roots and are destroyed, the hollow root with key 8 loses one parent.  All links are ranked. }\label{fig:hollow-heap}
\end{figure}

\begin{theorem}\label{thm:lazy-correct}
Two-parent hollow heaps correctly implement all the heap operations and maintain each heap as a heap-ordered dag with one full root and no empty roots.
\end{theorem}

\begin{proof}
The proof is straightforward by induction on the number of heap operations.
\end{proof}

The analysis of two-parent hollow heaps is less straightforward than that of one-parent hollow heaps.  To do the analysis, we define a way of virtually moving children among nodes that mimics their movement in one-parent hollow heaps.  This definition is part of the analysis only; it is \emph{not} part of the algorithm.  First we prove a few properties of the algorithm.

\begin{lemma}\label{lem:4.2}
A node in a two-parent hollow heap has at most one parent if it is full, at most two if it is hollow.  Once a node is hollow it cannot acquire a new parent.
\end{lemma}

\begin{proof} There are only two ways for a node to acquire a parent: a full root can acquire a parent by losing a \link, and a full node can acquire a parent by becoming hollow in a \decreasekey\ operation.  A root has no parents, so when a node loses a \link\ it has only one parent.  Once a node becomes hollow, it cannot become full again, so it cannot acquire a new parent.
\end{proof}


We call nodes that hold the same item (at different times) \emph{clones}.

\begin{lemma}\label{lem:4.3} 
At any given time, any maximal set of clones forms a path in the dag, in decreasing order by creation time.  Each clone on the path, except possibly the first, is hollow.  The path changes only by insertion or deletion at the front.  Thus the path forms a stack.
%
%
%
\end{lemma}

\begin{proof}
The proof is by induction on the number of heap operations.  Insertion creates a new path of clones consisting of one full node.  A \decreasekey\ that moves an item from a node~$u$ to a node~$v$ makes~$u$ hollow and adds~$v$ to the front of the path of clones of~$u$.  A \hdelete\ or \deletemin\ begins by removing an item from a node that is first on a path of its clones, making the node hollow.  Subsequently, the path can change only by destruction of its nodes, from front to back.
%
\end{proof}

Now we are ready to describe the virtual movement of children and present a rank invariant for two-parent hollow heaps.  We give nodes \emph{virtual parents} as follows: When a node is created, it has no virtual parent.  When a root with no virtual parent loses a ranked link, its new parent becomes its virtual parent.  When the virtual parent of a node is destroyed, it loses its virtual parent.  (If the node is full, it can acquire a new virtual parent later, by losing a link.)  When a \decreasekey\ moves an item from a node~$u$ to a node~$v$, we make~$v$ the virtual parent of all the nodes whose virtual parent is~$u$ and whose rank is less than $v.rank = \max\{0, u.rank - 2\}$.  A node is a \emph{virtual child} of its virtual parent.

\begin{lemma}\label{lem:4.4}
A root has no virtual parent.
\end{lemma}

\begin{proof}
The proof is by induction on the number of heap operations.  A newly created node has no virtual parent.  Suppose a node~$w$ acquires a virtual parent by losing a \link\ to~$u$.  Subsequently the virtual parent of~$w$ can change, but only to a later clone of~$u$.  For~$w$ to become a root again, $u$ must be destroyed, but by Lemma~\ref{lem:4.3} this can only happen after all later clones of~$u$ are destroyed, including the most recent virtual parent of~$w$.  Hence when~$w$ becomes a root again, it has no virtual parent.
\end{proof}

The rank invariant for two-parent hollow heaps is:

\begin{quote}
 A node~$u$ of rank $r$ has exactly $r$ virtual children, of ranks $0, 1,\ldots, r - 1$, unless $r > 2$ and~$u$ was made hollow by a \decreasekey, in which case~$u$ has exactly two virtual children, of ranks $r - 2$ and $r - 1$.
\end{quote}

\begin{theorem}\label{thm:4.5}
Two-parent hollow heaps maintain the rank invariant.
\end{theorem}

\begin{proof}
The proof is an induction on the number of heap operations like the proofs of the corresponding parts of Theorems~\ref{thm:multi-correct} and~\ref{thm:eager-correct}, using Lemma~\ref{lem:4.4}: by the lemma, when a node acquires a child by a ranked link it becomes the virtual parent of that child.
\end{proof}

Since a node has at most one virtual parent at a time, the virtual parents define a set of rooted trees that we call \emph{virtual trees}.  A node~$w$ is a \emph{virtual descendant} of a node~$v$ if $w = v$ or~$w$ is a virtual descendant of a virtual child of~$v$.

\begin{lemma}\label{lem:lazy-fib} 
A node of rank $r$ in a two-parent hollow heap has at least $F_{r + 3} - 1$ virtual descendants, both full and hollow.
\end{lemma}

\begin{proof}
The proof is like that of Lemma~\ref{lem:multi-fib}.
\end{proof}

\begin{corollary}\label{cor:lazy-rank} 
The rank of a node in a two-parent hollow heap of $N$ nodes is at most $\log_\phi N$.
\end{corollary}

\begin{proof}
The proof is the same as that of Corollary~\ref{cor:multi-rank}, using Lemma~\ref{lem:lazy-fib}.
\end{proof}

A straightforward way to implement two-parent hollow heaps is to store each set of children in an exogenous singly-linked list.  The order of children does not matter.  A \link\ makes the loser the first child of the winner.  The lists must be exogenous rather than endogenous since nodes can be on two lists at the same time.  We give each node a bit that is true if it has two parents.  The bit is initially false.  It becomes true when a \decreasekey\ makes the node hollow, unless the node is currently the root.  When a node is destroyed, each of its children with a false bit becomes a root; each of its children with a true bit has it set to false.  In Section~\ref{sec:lazy-implementation} we describe alternative implementations that use less space.

\begin{theorem}\label{thm:lazy-analysis} 
The amortized time per two-parent hollow heap operation is $O(1)$ for each operation other than a \hdelete\ or \deletemin, and $O(\log N)$ per \hdelete\ or \deletemin\ on a heap of $N$ nodes.
\end{theorem}

\begin{proof} The proof is like the proof of Theorem~\ref{thm:multi-analysis} but with two changes.  A hollow node can lose a parent without becoming a root and being destroyed (because it has another parent).  We charge each such event to the \decreasekey\ that gave the hollow node its second parent, at most one per \decreasekey.  To count ranked links, we use a slightly different potential function: each full node has one unit of potential unless it is a full \emph{virtual} child of a full node.  A ranked link reduces the potential by one, giving it an amortized cost of zero.  The unranked links done in melds do not change the potential.  Making a node hollow by removing its item increases the potential by at most $\log_\phi N$, since only virtual children increase in potential.
\end{proof}

We conclude this section by comparing two-parent and one-parent hollow heaps.  Consider a \decreasekey\ that moves an item from $u$ to a new node $v$.  Eventually one of $u$ and $v$ may be destroyed, but the one destroyed first depends on future heap operations, which are not known in advance.  All children of the node destroyed become roots, and must themselves be destroyed (if hollow) or melded (if full).  In one-parent hollow heaps, no matter how the children of $u$ are distributed between $u$ and $v$ when the \decreasekey\ occurs, in the worst case at least half of them will become roots when the first of $u$ and $v$ is destroyed.  In a two-parent hollow heap, on the other hand, none of the children of $u$ become roots until both $u$ and $v$ are destroyed.  Thus not only do two-parent hollow heaps save work by not moving children, they delay the next access to these children by postponing when they become roots.  We leave as an open question whether this local advantage translates into a global advantage.

\section{Rebuilding}\label{sec:rebuilding}

The number of nodes $N$ in a hollow heap is at most the number of insertions plus the number of \decreasekey\ operations on items that were ever in the heap or in heaps melded into it.  If the number of \decreasekey\ operations is polynomial in the number of insertions, and only items in minimum nodes are deleted, then $\log N = O(\log n)$, where $n$ is the number of heap items, so the amortized time per \hdelete\ or \deletemin\ is $O(\log n)$, the same as for Fibonacci heaps.  In applications in which the storage required for the problem input is at least linear in the number of heap operations, the extra space needed for hollow nodes is at most linear in the problem size.  Both of these conditions hold for the heaps used in many graph algorithms, including Dijkstra's shortest path algorithm \cite{Di59,FrTa87}, various minimum spanning tree algorithms \cite{Di59,FrTa87,GaGaSpTa86,Prim57}, and Edmonds' optimum branching algorithm \cite{Edmonds67,GaGaSpTa86}.  In these applications there is at most one \hinsert\ and one \deletemin\ per vertex, no \hdelete\ operations other than those triggered by \deletemin\ operations, and at most two \decreasekey\ operations per edge or arc.  Furthermore the number of edges or arcs is at most quadratic in the number of vertices.  In these applications hollow heaps are asymptotically as time-efficient as Fibonacci heaps, although they need space linear in the number of edges or arcs, whereas Fibonacci heaps only need space linear in the number of vertices.  Whether the space advantage of Fibonacci heaps is significant depends on the graph density and on implementation details.

For applications in which the number of \decreasekey\ operations is huge compared to the number of insertions, or the extra space needed by hollow nodes becomes a bottleneck, we can use periodic rebuilding to guarantee that $N = O(n)$ for every heap.  To do this, keep track of $N$ and $n$ for every heap.  When $N > cn$ for a suitable constant $c > 1$, eliminate all hollow nodes by rebuilding the heap.

We offer two ways to do the rebuilding.  The first is to completely disassemble the heap and build a new one containing only the full nodes, as follows: Destroy every hollow node.  Make each full node into a one-node heap whose node has rank $0$.  Do repeated melds until one heap remains.   A second method that does no key comparisons is to give all the full nodes a rank of $0$ and “contract” all the hollow nodes, as follows: In a two-parent hollow heap, eliminate one parent of every node that has two, making the dag a tree (or, in the multi-root case, a forest).  Give each full node a rank of zero and give each full child a parent equal to its nearest full proper ancestor.  Destroy all the hollow nodes.  To extend the analysis in Sections~\ref{sec:multi-root}-\ref{sec:lazy} to cover the second rebuilding method, we define every child to be unranked after the rebuilding.  Either way of rebuilding can be done in a single traversal of the dag (or tree or set of trees), taking $O(N)$ time.  Since $N > cn$ and $c > 1$, $O(N) = O(N - n)$.  That is, the rebuilding time is $O(1)$ per hollow node destroyed.  By charging the rebuilding time to the heap operations that created the hollow nodes, $O(1)$ per operation, we obtain the following theorem:

\begin{theorem}\label{thm:rebuilding} 
With rebuilding, the amortized time per hollow heap operation is $O(1)$ for each operation other than a \deletemin\ or \hdelete, and $O(\log n)$ per \deletemin\ or \hdelete\ on a heap of $n$ items.  The space required by a heap is $O(n)$.  These bounds hold for both two-parent and one-parent hollow heaps.
\end{theorem}

By making $c$ sufficiently large, we can arbitrarily reduce the rebuilding overhead, at a constant factor cost in space and an additive constant cost in the amortized time of \hdelete.  Whether rebuilding is actually a good idea in any particular application is a question to be answered by experiments.

With rebuilding as described above, the worst-case time of a \decreasekey\ operation is no longer $O(1)$, since such a \decreasekey\ can trigger a rebuilding.  (The amortized time remains $O(1)$.)  One can preserve the $O(1)$ worst-case time of \decreasekey\ operations by doing \emph{incremental} rebuilding, although the details get a bit complicated.  A simpler alternative is as follows:  Use a multi-root version of hollow heaps (either one-parent or two-parent).  During a \decreasekey, if the item whose key decreases is in a root, decrease the key of this node and update the minimum node, instead of moving the item to a new node.  Rebuild when a \hdelete\ or \deletemin\ causes $N > cn$, where $c > 1$.  With this method, at most one hollow node per item is created between consecutive \hdelete\ or \deletemin\ operations.  Thus if $N \le cn$ just after one such operation, $N \le (c + 1)n$ just before the next one.  By choosing $c$ appropriately, one can guarantee that the fraction of nodes that are hollow is always at most $\frac{1}{2}-\varepsilon$, for any fixed $\varepsilon > 0$.

A natural question to ask is whether hollow nodes can be entirely eliminated from the data structure.  The answer is yes, at least for one-parent hollow heaps.  The idea is to ``contract'' hollow nodes as soon as they are created.  We shall describe the application of this idea to the multi-root hollow heaps of Section~\ref{sec:multi-root}, modified as described there to keep all roots full.  In such a hollow heap, let $u$ be a hollow child with full parent $v$.  Let $T$ be the maximal subtree of hollow nodes rooted at $u$, and let $C$ be the list of (full) children of the leaves of $T$, in symmetric order (as defined by the order of nodes in lists of children).  In the corresponding contracted structure, $T$ is deleted and $u$ is replaced  as a child of $v$ by list $C$, which becomes a sublist of children of $v$. 
For this idea to work, we need to keep track of such sublists of children, since it is these sublists, rather than individual children, that are moved during \decreasekey\ operations.

Here are the (somewhat tedious) details.  As in Section~\ref{sec:multi-root}, maintain each list of children in decreasing order by rank, but partition each such list into sublists of consecutive children.  The behavior of the algorithm defines these sublists.  When a ranked link makes the loser~$u$ the first child of the winner~$v$, $u$ becomes the only member of a singleton sublist that is first among the sublists of children of~$v$; the remaining sublists of children of~$v$ remain the same.

A \decreasekey\ does not create a new node but merely changes the key of the old node holding the item and moves appropriate sublists of children.  To do a \decreasekey\ on the item in a node~$u$, update the key of~$u$.  If~$u$ is a root, this completes the \decreasekey.  If not, let~$v$ be the parent of~$u$ and~$A$ the sublist of children of~$v$ containing~$u$.  Split~$A$ into two sublists~$A'$ and~$A''$, with~$u$ first on~$A''$.  Sublist~$A'$ may be empty; $A''$ may contain only~$u$.  Delete~$u$ from $A''$.  Form sublist~$B$ as follows: if $u.rank = 0$, let~$B$ be the empty list; if $u.rank = 1$, remove from~$u$ its first sublist of children and let this sublist be~$B$; if $u.rank > 1$, remove from~$u$ its first two sublists of children and catenate them to form~$B$.  Set $v.rank = \max \{0, u.rank - 2 \}$.  Replace~$A$ in the list of sublists of children of~$v$ by the catenation of~$A'$, $B$, and~$A''$ in this order.  Add $u$ (with its remaining sublists of children) to the list of roots, and update the minimum node.

We call this data structure the {\rm contracted hollow heap}.  The analysis of Section~\ref{sec:multi-root} extends to give the following theorem:

\begin{theorem}  The amortized time per contracted hollow heap operation is $O(1)$ for each operation other than a \deletemin\ or $\hdelete$, and $O(\log N)$ per \deletemin\ or \hdelete\ on a heap formed by $N$ \hinsert\ and \decreasekey\ operations and any number of \meld\ operations.
\end{theorem}

It is straightforward to extend contraction to one-root one-parent hollow heaps, but we do not know how to extend it to two-parent hollow heaps in a nice way: contraction can cause a node to have an arbitrarily large number of parents.

Contracted hollow heaps contain no hollow nodes, but they have at least two significant drawbacks.  The time per \hdelete\ is $O(\log N)$, not $O(\log n)$: to obtain the latter, some form of additional rebuilding seems necessary.  The number of pointers per node to implement the structure is large (exceeding that of Fibonacci heaps and rank-pairing heaps), although the structure can be made endogenous, saving some pointers.  Thus we view contracted hollow heaps as primarily of theoretical interest.

\section{Implementation of two-parent hollow heaps}\label{sec:lazy-implementation}
As mentioned in Sections \ref{sec:multi-root} and \ref{sec:one-root}, one can implement one-parent hollow heaps (either multi-root or one-root) using a pointer per item (to the node holding it) and a rank and three pointers per node (to the first child, next sibling, and the item it holds).  The implementation of two-parent hollow heaps uses more space, since each list of children must be exogenous rather than endogenous.  In this section we describe implementations of two-parent hollow heaps that represent the lists of children endogenously and thereby reduce the space needed.

If~$v$ is a parent of~$u$, we say that~$v$ is the \emph{first} or \emph{second parent} of~$u$ if~$u$ acquired parent~$v$ via a \link\ or a $\decreasekey$, respectively.  Only a hollow node can have two parents; it can lose them in either order.

The way we have defined links to work in Section \ref{sec:lazy} is the key to making the lists of children endogenous.  A link makes the loser the \emph{first} child of the winner.  This guarantees that a node~$u$ with two parents is always the \emph{last} child of its second parent~$v$: when~$u$ becomes a child of~$v$,~$u$ is the only child of~$v$, and any children that~$v$ later acquires are added in front of~$u$ on the list of children of~$v$.

This allows us to use two pointers per node to represent lists of children, as in one-parent hollow heaps: if~$u$ is a node, $u.child$ is the first child of~$u$, $\NULL$ if~$u$ has no children; if~$u$ is a child, $u.next$ is the next sibling of~$u$ on the list of children of its first parent, $\NULL$ if it is the last child of its first parent.

With this representation, given a child~$u$ of a node~$v$, we need ways to answer three questions: (i) Is~$u$ last on the list of children of~$v$? (ii) Does~$u$ have two parents? (iii) Assuming~$u$ has two parents, is~$v$ the first or the second?  There are several methods that allow these questions to be answered in $O(1)$ time.  We provide a detailed implementation of two-parent hollow heaps using one method and discuss alternatives below.

Each node~$u$ has a pointer $u.item$ to the item in~$u$ if~$u$ is full, $\NULL$ if~$u$ is hollow.  Each node~$u$ has another pointer $u.ep$ (for “extra parent”) that is the second parent of~$u$ if~$u$ has two parents, $\NULL$ if~$u$ has at most one.  In particular, $u.ep = \NULL$ if~$u$ is full.  As an optimization, a \decreasekey\ on the item in the minimum node does not create a new node but merely changes the key of the minimum node.  A decrease-key on an item not in the minimum node makes the newly hollow node~$u$ a child of the new full node~$v$ by setting $v.child = u$ and $u.ep = v$ but \emph{not} changing $u.next$: $u.next$ is the next sibling of~$u$ on the list of children of the first parent of~$u$.

We answer the three questions as follows: (i) A child~$u$ of~$v$ is last on the list of children of~$v$ if and only if $u.next = \NULL$ ($u$ is last on any list of children containing it) or $u.ep = v$ ($u$ is hollow with two parents and~$v$ is its second parent); (ii) $u$ has two parents if and only if $u.ep \ne \NULL$; and (iii) Assuming~$u$ has two parents, $v$ is the second if and only if $v = u.ep$.

Each node~$u$ also stores its key and rank, and each item~$e$ stores the node $e.node$ holding it. The total space needed is four pointers, a key, and a rank per node, and one pointer per item.  Ranks are small integers, each requiring $\lg\lg N + O(1)$ bits of space.

Implementation of all the heap operations except \hdelete\ is straightforward.  Figure~\ref{fig:5.1} gives such implementations in pseudocode; Figure~\ref{fig:5.2} gives implementations of auxiliary methods used in  Figure~\ref{fig:5.1}.

\newcommand{\MAKEHEAP}{
\parbox[t]{2.5in}{
\begin{mytabbing}
aaa\=aaa\=aaa \kill
$\makeheap()$: \\
\>      \RETURN\ $\NULL$
\end{mytabbing}
}}

\newcommand{\FINDMIN}{
\parbox[t]{2.5in}{
\begin{mytabbing}
aaa\=aaa\=aaa \kill
$\findmin(h)$: \\
\>     \IF\ $h = \NULL$: \RETURN\ $\NULL$ \\
\>     \ELSE: \RETURN\ $h.\ITEM$
\end{mytabbing}
}}

\newcommand{\INSERT}{
\parbox[t]{2.5in}{
\begin{mytabbing}
aaa\=aaa\=aaa \kill
$\hinsert(e, k, h)$: \\
\>      \RETURN\ $meld(\makenode(e,k), h)$
\end{mytabbing}
}}

\newcommand{\MELD}{
\parbox[t]{2.5in}{
\begin{mytabbing}
aaa\=aaa\=aaa \kill
$\meld(g, h)$: \\
$\{$ \>    \IF\ $g = \NULL$: \RETURN\ $h$ \\
\>      \IF\ $h = \NULL$: \RETURN\ $g$ \\
\>      \RETURN\ $link(g, h)$   $\;$ $\}$
\end{mytabbing}
}}

\newcommand{\DK}{
\parbox[t]{2.5in}{
\begin{mytabbing}
aaa\=aaa\=aaa \kill
$\decreasekey(e, k, h):$ \\
$\{$  \>      $u = e.\node$   \\
\>       \IF\ $u = h$: \\
\>        $\{$ \> $u.\key = k$ \\
\>             \> \RETURN\ $h$ $\}$ \\
\>       $v = \makenode(e,k)$ \\
\>       $u.\ITEM = \NULL$ \\
\>       \IF\ $u.\rank > 2$: $v.\rank = u.\rank - 2$ \\
\>       $v.\child = u$ \\
\>       $u.ep = v$ \\
\>       \RETURN\ $\link(v, h)$   $\;$ $\}$
\end{mytabbing}
}}

\newcommand{\DELETEMIN}{
\parbox[t]{2.5in}{
\begin{mytabbing}
aaa\=aaa\=aaa \kill
$\deletemin(h)$: \\
\>      \RETURN\ $\hdelete(h.\ITEM, h)$
\end{mytabbing}
}}

\newcommand{\MAKENODE}{
\parbox[t]{2.5in}{
\begin{mytabbing}
aaa\=aaa\=aaa \kill
$\makenode(e,k)$: \\
$\{$ \>     $u = \newnode( )$ \\
\>       $u.\ITEM = e$    \\
\>       $e.\node = u$     \\
\>       $u.\child = \NULL$ \\
\>       $u.\next = \NULL$ \\
\>       $u.ep = \NULL$ \\
\>       $u.key = k$  \\
\>       $u.\rank = 0$  \\
\>       \RETURN\ $u$ $\;$    $\}$ \\
\end{mytabbing}
}}

\newcommand{\LINK}{
\parbox[t]{2.5in}{
\begin{mytabbing}
aaa\=aaa\=aaa \kill
$\link(v, w)$: \\
\>      \IF\ $v.\key \ge w.\key$: \\
\>      $\{$  \>  $\addchild(v, w)$       \\
\>            \> \RETURN\ $w$    $\}$     \\
\>      \ELSE:                             \\
\>      $\{$  \>  $\addchild(w, v)$           \\
\>            \> \RETURN\ $v$    $\}$
\end{mytabbing}
}}

\newcommand{\ADDCHILD}{
\parbox[t]{2.5in}{
\begin{mytabbing}
aaa\=aaa\=aaa \kill
$\addchild(v, w)$: \\
$\{$  \>    $v.\next = w.\child$ \\
       \>   $w.\child = v$  $\;$   $\}$
\end{mytabbing}
}}

\begin{figure}[htbp]
\begin{center}
\parbox[t]{2.3in}{
\MAKEHEAP\\
\INSERT\\
\MELD\\
\FINDMIN
}
\parbox[t]{2.3in}{
\DK\\
\DELETEMIN
}
\end{center}
\caption{Implementations of all two-parent hollow heap operations except \hdelete.}
\label{fig:5.1}
\end{figure}

\begin{figure}[htbp]
\begin{center}
\parbox[t]{2.3in}{
\MAKENODE
}
\parbox[t]{2.3in}{
\LINK\\
\ADDCHILD
}
\end{center}
\caption{Implementations of auxiliary methods used in Figure~\ref{fig:5.1}.}
\label{fig:5.2}
\end{figure}

Implementation of \hdelete\ requires keeping track of roots as they are destroyed and linked. To do this, we maintain a list $L$ of hollow roots, singly linked by \next\ pointers.  We also maintain an array $A$ of full roots, indexed by rank, at most one per rank.

\begin{figure}[htbp]
\begin{center}
\parbox[t]{5in}{
\begin{mytabbing}
aaa\=aaa\=aaa\=aaa\=aaa\=aaa\=aaa \kill
$\hdelete(e, h)$: \\
$\{$ \> $e.\node.\ITEM = \NULL$ \\
\>      $e.\node = \NULL$ \\
\>      if $h.\ITEM\not= \NULL$: \RETURN\ $h$ \hspace{1cm} /* Non-minimum deletion */ \\
\>      $\maxrank = 0$ \\
\>      \WHILE\ $h\not= \NULL$: \hspace{1cm} /* While $L$ not empty */ \\
\>      $\{$ \>    $w = h.\child$ \\
\>           \>    $v = h$     \\
\>           \>    $h = h.\next$ \\
\>           \>    \WHILE\ $w\not= \NULL$:   \\
\>           \>    $\{$  \>     $u = w$ \\
\>           \>          \>     $w = w.\next$ \\
\>           \>          \>     \IF\ $u.\ITEM = \NULL$: \\
\>           \>          \>        $\{$   \>  \IF\ $u.ep=\NULL$: \\
\>           \>          \>               \>  $\{$   \>  $u.\next = h$   \\
\>           \>          \>               \>         \>  $h = u$  $\;\}$  $\;\}$   \\
\>           \>          \>               \>  \ELSE:   \\
\>           \>          \>               \>  $\{$ \> \IF\ $u.ep=v$: $w=\NULL$    \\
\>           \>          \>               \>       \> \ELSE: $u.next = null$   \\
\>           \>          \>               \>       \>  $u.ep = null$  $\}$   \\
\>           \>          \>     \ELSE:  \\
\>           \>          \>               \>       $\dorankedlinks(u)$ \\
\>           \>     destroy $v$ $\;\}$ \\
\>      $\dounrankedlinks()$\\
\>      \RETURN\ $h$    $\}$
\end{mytabbing}
}
\end{center}
\caption{Implementation of \hdelete\ in two-parent hollow heaps. Rank updates during ranked links are done in the auxiliary method \dorankedlinks\ in Figure~\ref{fig:delete-auxiliary}.
\label{fig:tricky-lazy}}
\end{figure}

\newcommand{\DORANKEDLINKS}{
\parbox[t]{2.5in}{
\begin{mytabbing}
aaa\=aaa\=aaa \kill
$\dorankedlinks(u)$: \\
$\{$   \>  \WHILE\ $A[u.\rank]\not= \NULL$: \\
       \>  $\{$  \>  $u = \link(u, A[u.\rank])$  \\
       \>        \>  $A[u.\rank] = \NULL$  \\
       \>        \>  $u.\rank = u.\rank + 1$     $\;\}$ \\
       \>  $A[u.\rank] = u$   \\
       \>  \IF\ $u.\rank > \maxrank$: \\
       \>        \>  $\maxrank = u.\rank$  $\;\}$  $\;\}$
\end{mytabbing}
}}

\newcommand{\DOUNRANKEDLINKS}{
\parbox[t]{2.5in}{
\begin{mytabbing}
aaa\=aaa\=aaa \kill
$\dounrankedlinks()$: \\
\FOR\  $i = 0$ to $\maxrank$:  \\
$\{$ \>    \IF\ $A[i]\not= \NULL$: \\
     \>    $\{$  \>   \IF\ $h = \NULL$: $h = A[i]$ \\
     \>          \>   \ELSE: $h = \link(h, A[i])$ \\
     \>          \>   $A[i] = \NULL$    $\;\}$     $\;\}$
\end{mytabbing}
}}

\begin{figure}[htbp]
\begin{center}
\parbox[t]{2.3in}{
\DORANKEDLINKS
}
\parbox[t]{2.3in}{
\DOUNRANKEDLINKS
}
\end{center}
\caption{Implementations of auxiliary methods used in $\hdelete$.}
\label{fig:delete-auxiliary}
\end{figure}

When a \hdelete\ makes a root hollow, do the following.  First, initialize $L$ to contain the hollow root and $A$ to be empty.  Second, repeat the following until $L$ is empty: Delete a node~$v$ from $L$, apply the appropriate one of the following cases to each child~$u$ of~$v$, and then destroy~$v$:


\begin{itemize}
\item[(a)]  $u$ is hollow and~$v$ is its only parent: Add~$u$ to $L$: deletion of~$v$ makes~$u$ a root.
\item[(b)] $u$ has two parents and~$v$ is the second: Set $u.ep = \NULL$ and stop processing children of~$v$: $u$ is the last child of~$v$.
\item[(c)]    $u$ has two parents and~$v$ is the first: Set $u.ep = \NULL$ and $u.next = \NULL$.
\item[(d)]   $u$ is full: Add~$u$ to $A$ unless $A$ contains a root of the same rank. If it does, link~$u$ with this root via a ranked link and repeat this with the winner until $A$ does not contain a root of the same rank; then add the final winner to $A$.
\end{itemize}

Third and finally (once $L$ is empty), empty $A$ and link full roots via unranked links until there is at most one.

Figure~\ref{fig:tricky-lazy} gives pseudocode that implements \hdelete.  Since cases (ii) and (iii) both set $u.ep = \NULL$, this assignment is factored out of these cases. Figure~\ref{fig:delete-auxiliary} gives auxiliary methods used by \hdelete\ to do links.
Array $A$ is a global variable, assumed to be initialized to empty. Integer $\maxrank$ is also a global variable.

With this implementation, the worst-case time per operation is $O(1)$ except for \hdelete\ operations that remove root items.  A \hdelete\ that removes a root item takes $O(1)$ time plus $O(1)$ time per hollow node that loses a parent plus $O(1)$ time per link plus $O(\log_\phi N)$ time, where $N$ is the number of nodes in the dag just before the \hdelete, since $\maxrank = O(\log_\phi N)$ by Corollary~\ref{cor:lazy-rank}. These are the bounds needed to give Theorem~\ref{thm:lazy-analysis}.

We can reduce the number of pointers per node in this implementation from four to three by using the same field of a node~$u$ to hold $u.item$ and $u.ep$, since $u.ep$ is $\NULL$ if~$u$ is full and $u.item$ is $\NULL$ if~$u$ is hollow.  This requires adding a bit per node to indicate whether the node is full or hollow, trading a bit per node for a pointer per node.  We can avoid the extra bit per node by using the rank field to indicate hollow nodes, for example by setting the rank of a hollow node to 0 and that of a full node to its actual rank plus one.  This variant has the disadvantage that the shared field for items and extra parents must be able to store pointers to two different types of objects.

An alternative trades three bits per node for one pointer per node but does not require fields storing pointers of different types: Eliminate $ep$ pointers.  Instead, store with each node~$u$ three Boolean variables; $u.new$, true if and only if~$u$ was created by a \decreasekey\ and has only one child (the node whose item was moved to~$u$); $u.penult$ (penultimate), true if and only if~$u$ is the next-to-last child of its first parent; and $u.two$, true if and only if~$u$ has two parents.  Answer questions (i) and (ii) as follows: (i) $u$ is last on the list of children of~$v$ if $v.new$ or $x.penult$, where~$x$ is the child before~$u$ on the list of children of~$v$ (visited just before~$u$ during the traversal of the list); and (ii) $u$ has two parents if and only if $u.two$.  Question (iii) does not need to be answered: during a delete, when processing each child~$u$ of a hollow root~$v$, there are only three cases, not four:

\begin{itemize}
\item[(a')] Same as Case (a).
\item[(b')] Combines Cases (b) and (c): $u$ has two parents: Set $u.two = false$.
\item[(c')] Same as Case (d).
\end{itemize}

Another alternative eliminates~$ep$ pointers in a different way: delete the~$ep$ pointers, store a Boolean variable $u.two$ with each node~$u$ that is true if and only if~$u$ has two parents, and modify the item field for a hollow node~$u$ to point to the item~$u$ once held.  This alternative requires that each deleted item~$e$ has $e.node = \NULL$; otherwise, hollow nodes can have dangling item pointers.  With this method, a node~$u$ is full if and only if $u.item.node = u$, a child~$u$ of~$v$ is last on the list of children of~$v$ if and only if $u.next = \NULL$ or $u.item = v.item$, and~$v$ is the second parent of~$u$ if and only if $u.item = v.item$.  Since item pointers replace~$ep$ pointers, an item cannot be destroyed until all hollow nodes previously containing it have been destroyed.  This alternative is especially appealing if keys and ranks are stored with items instead of nodes.  Then one can do a \decreasekey\ operation by accessing only the item whose key decreases, not the node~$u$ containing it.  In particular, creating a new node~$v$ and setting $e.item = v$ automatically makes u hollow since then $u.item.node = v$.

Which alternative is best in practice is an experimental question and is likely to depend on the sequence of operations as well as the details of the computer and programming language.

\section{Good and Bad Variants}\label{sec:good-bad}

In this section we explore the design space of hollow heaps.  We show that our data structures occupy ``sweet spots'' in the design space: although small changes to these  structures preserve their efficiency, larger changes destroy it.
We explore variants of the one-root structures of Sections~\ref{sec:one-root} and~\ref{sec:lazy}; our results extend to the analogous multi-root
structures as well.
We consider three classes of structures: \lazy-$k$, \eager-$k$, and \naive-$k$.  Here~$k$ is an integer function specifying the rank of the new node~$v$ in a \decreasekey\ operation as a function of the rank~$r$ of the node~$u$ made hollow by the operation. Data structure \lazy-$k$ is the data structure of Section~\ref{sec:lazy}, except that it sets the rank of~$v$ in \decreasekey\ to be $\max\{k, 0\}$.  Thus \lazy-$(r - 2)$ is exactly the data structure of Section~\ref{sec:lazy}.  Data structure \eager-$k$ is the data structure of Section~\ref{sec:one-root}, except that it sets the rank of~$v$ in \decreasekey\ to be $\max\{k, 0\}$,
and, if $r > k$, it moves to~$v$ all but the~$r-k$ highest-ranked ranked children of~$u$, as well as the unranked children of~$u$.
Thus \eager-$(r - 2)$ is exactly the data structure of Section~\ref{sec:one-root}.  Finally, \naive-$k$ is the data structure of Section~\ref{sec:lazy}, except that it sets the rank of~$v$ in \decreasekey\ to be $\max\{k, 0\}$ and it never assigns second parents: when a hollow node~$u$ becomes a root, $u$ is deleted and all its children become roots.  We consider two regimes for~$k$: \emph{large}, in which $k = r - j$ for some fixed non-negative integer~$j$; and \emph{small}, in which $k = r - f(r)$, where $f(r)$ is a positive non-decreasing integer function that tends to infinity as $r$ tends to infinity.

We begin with a positive result: for any fixed integer $j \ge 2$, both \lazy-$(r - j)$ and \eager-$(r - j)$ have the efficiency of Fibonacci heaps.  It is straightforward to prove this by adapting the analysis in Sections~\ref{sec:one-root} and~\ref{sec:lazy}.  As~$j$ increases, the rank bound (Corollaries~\ref{cor:eager-rank} and~\ref{cor:lazy-rank}) decreases by a constant factor, approaching $\lg N$ or $\lg n$, respectively, as~$j$ grows, where $\lg$ is the base-2 logarithm.  The trade-off is that the amortized time bound for \decreasekey\ is $O(j + 1)$, increasing linearly with~$j$.

All other variants are inefficient.  Specifically, if the amortized time per \deletemin\ is $O(\log m)$, where $m$ is the total number of operations, and the amortized time per \makeheap\ and \hinsert\ is  $O(1)$, then the amortized time per \decreasekey\ is $\omega(1)$.  We demonstrate this by constructing costly sequences of operations for each variant.  We content ourselves merely with showing that the amortized time per \decreasekey\ is $\omega(1)$; for at least some variants, there are asymptotically worse sequences than ours. Our results are summarized in the following theorem. The individual constructions appear in Sections \ref{sec:eager-k}, \ref{sec:r_and_r-1}, and \ref{sec:lazy-k}.

\begin{theorem}\label{thm:7.1}
Variants \lazy-$(r - j)$ and \eager-$(r - j)$ are efficient for any choice of $j > 1$ fixed independent of $r$.
All other variants, namely \naive-$k$
for all~$k$, \eager-$r$, \lazy-$r$, \eager-$(r - 1)$, \lazy-$(r - 1)$, and  \eager-$k$ and \lazy-$k$ for~$k$ in the small regime
are inefficient.
\end{theorem}

\subsection{\eager-$k$ for $k$ in the small regime and \naive-$k$ for all $k$}\label{sec:eager-k}

We first consider \eager-$k$ for $k$ in the small regime, i.e., $k = r - f(r)$ where $f$ is a positive non-decreasing function that tends to infinity.
We obtain an expensive sequence of operations as follows.  We define the \emph{binomial tree}~$B_\ell$ \cite{Brown78,Vuillemin78} inductively: $B_0$ is a one-node tree; $B_{\ell+1}$ is formed by linking the roots of two copies of~$B_\ell$.  Tree~$B_\ell$ consists of a root whose children are the roots of copies of $B_0, B_1,\ldots, B_{\ell - 1}$ \cite{Brown78,Vuillemin78}.  For any~$\ell$, build a~$B_\ell$ by beginning with an empty tree and doing $2^\ell+1$ insertions of items in increasing order by key followed by one \deletemin.  After the insertions, the tree will consist of a root with $2^\ell$ children of rank~$0$.  In the \deletemin, all the links will be ranked, and they will produce a copy of~$B_\ell$ in which each node that is the root of a copy of $B_j$ has rank~$j$. The tree $B_\ell$ is shown at the top of Figure~\ref{fig:eager-k}.

\begin{figure}[t]
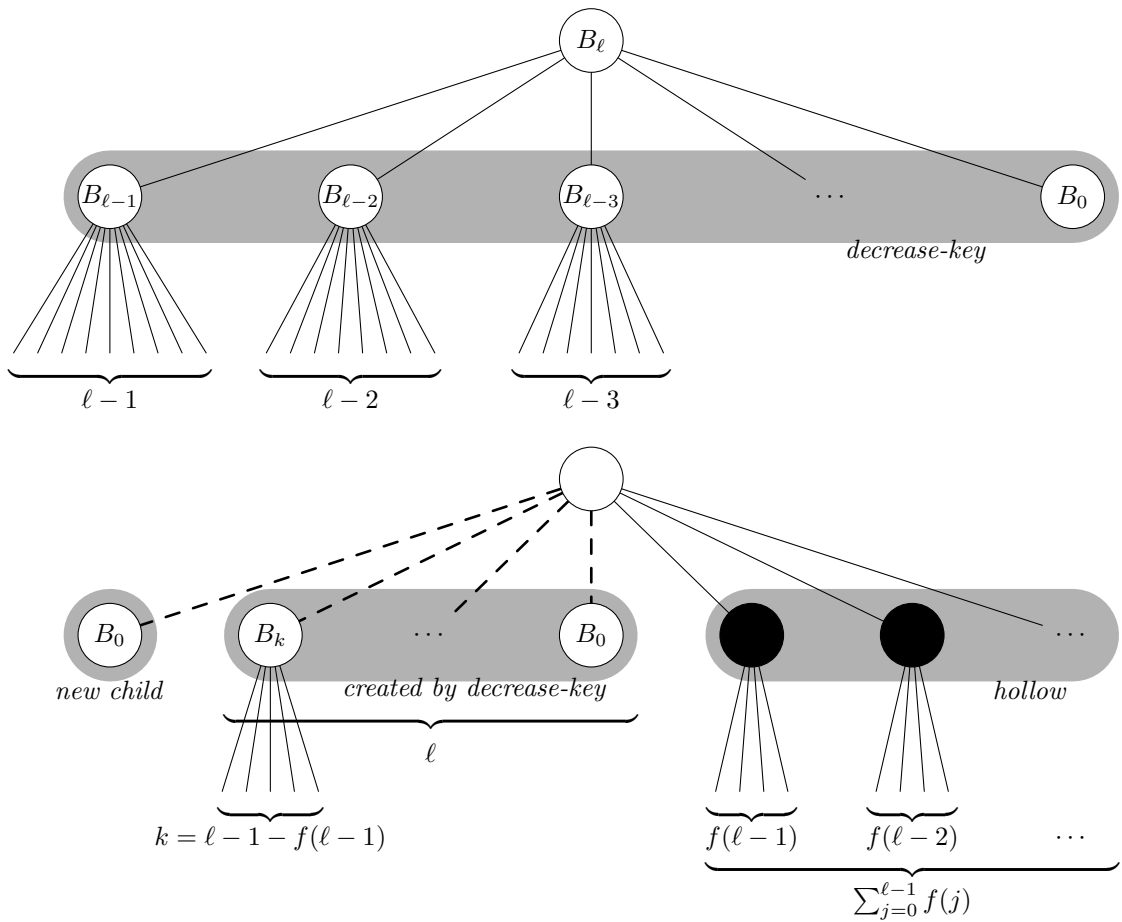

\begin{center}
\includegraphics[scale=1]{figures2.1}

\vspace{0.4cm}

\includegraphics[scale=1]{figures2.2}
\end{center}
\caption{
The construction for \eager-$k$. Roots of binomial trees are labeled, and black nodes are hollow. Solid and dashed lines denote ranked and unranked links, respectively. (Top) The initial configuration - a binomial tree $B_\ell$. The shaded region shows nodes on whose items \decreasekey\ operations are performed. (Bottom) The heap after performing \decreasekey\ operations and inserting a new child. The keys of the items in the newly hollow nodes were decreased, resulting in the middle nodes being linked with the root. The number of children of each node is shown at the bottom.
}\label{fig:eager-k}
\end{figure}

Now repeat the following $\ell+2$ operations $2^\ell$ times: do~$\ell$ \decreasekey\ operations on the items in the children of the root of~$B_\ell$, making the new keys greater than that of the key of the item in the root. This makes the $\ell$ previous children of the root hollow, and gives the root $\ell$ new children. Insert a new item whose key is greater than that of the item in the root.  Finally, do a \deletemin.  The \deletemin\ deletes the root and its~$\ell$ hollow children, leaving the children of the hollow nodes to be linked. Since a hollow node of rank $r$ has $f(r)$ children, the total number of nodes linked after the \deletemin\ is $1+\ell+\sum_{j=0}^{\ell-1} f(j) > (\ell/2)f(\ell/2)$. \footnote{Assume for simplicity that $\ell$ is even.} Each node involved in linking is the root of a binomial tree. Since the total number of nodes remains $2^\ell$, the binomial trees are linked using only ranked links to form a new copy of~$B_\ell$, and the process is then repeated.

After~$B_\ell$ is formed, each round of $\ell + 2$ consecutive subsequent operations contains only one \deletemin\ but takes $\Theta(\ell f(\ell/2))$ time. The total number of operations is $m=O(\ell2^\ell)$, of which $2^\ell + 1$ are \deletemin\ operations. The total time for the operations is $\Theta(\ell 2^\ell f(\ell/2))=\Theta(m f(\ell/2))$, but the desired time is $O(\ell2^\ell)=O(m)$. In particular, if the amortized time per \deletemin\ is $O(\ell)$ and the amortized time per \makeheap\ and \hinsert\ is $O(1)$, then the amortized time per decrease-key is $\Omega(f(\ell/2))$, which tends to infinity with~$\ell$.

We next consider \naive-$k$. Note that \naive-$0$ is identical to \eager-$0$, so the two methods do exactly the same thing for the example described above. An extension of the construction shows that \naive-$k$ is inefficient for every value of~$k$, provided that we let the adversary choose which ranked link to do when more than one is possible. Method \naive-$k$ is identical to \naive-$0$ except that nodes created by $\decreasekey$ may not have rank 0. The construction for \naive-$k$ deals with this issue by inserting new nodes with rank 0 that serve the function of nodes created by $\decreasekey$ for \naive-$0$. The additional nodes with non-zero rank are linked so that they do not affect the construction.

We build an initial~$B_\ell$ as before.  Then we do~$\ell$ \decreasekey\ operations  on the items in the children of the root, followed by $\ell + 1$ \hinsert\ operations of items with keys greater than that of the root, followed by one \deletemin\ operation, and repeat these operations $2^\ell$ times.  When doing the linking during the \deletemin, the adversary preferentially links newly inserted nodes and grandchildren of the deleted root, avoiding links involving the new nodes created by the \decreasekey\ operations until these are the only choices.  Furthermore, it chooses keys for the newly inserted items so that one of them is the new minimum.  Then the tree resulting from all the links will be a copy of~$B_\ell$ with one or more additional children of the root, whose descendants are the nodes created by the \decreasekey\ operations.  After the construction of the initial~$B_\ell$, each round of $2\ell+2$ subsequent operations maintains the invariant that the tree consists of a copy of~$B_\ell$ with additional children of its root, whose descendants are all the nodes added by \decreasekey\ operations.

The analysis is the same as for \eager-$0$, i.e. for the case  $f(r) = r$. The total number of operations is $m=O(\ell2^\ell)$, and the desired time is $O(\ell2^\ell)=O(m)$. The total time for the operations is however $\Theta(\ell^22^\ell)=\Theta(m\ell)$. Thus, the construction shows that \naive-$k$ for any value of~$k$ takes at least logarithmic amortized time per \decreasekey.



\subsection{\lazy-$r$, \eager-$r$, \lazy-$(r - 1)$, and \eager-$(r - 1)$}\label{sec:r_and_r-1}

Next we consider \lazy-$r$, \eager-$r$, \lazy-$(r - 1)$, and \eager-$(r - 1)$.  To get a bad example for each of these methods, we construct a tree~$T_\ell$ with a full root, having full children of ranks $0,1,\ldots, \ell - 1$, and in which all other nodes, if any, are hollow.  Then we  repeatedly do an \hinsert\ followed by a \deletemin, each repetition taking $\Omega(\ell)$ time.

In these constructions, all the \decreasekey\ operations are on nodes having only hollow descendants, so the operations maintain the invariant that every hollow node has only hollow descendants.  If this is true, the only effect of manipulating hollow nodes is to increase the cost of the operations, so we can ignore hollow nodes; or, equivalently, regard them as being deleted as soon as they are created.  Furthermore, with this restriction \lazy-$k$ and \eager-$k$ have the same behavior, so one bad example suffices for both \lazy-$r$ and \eager-$r$, and one for \lazy-$(r - 1)$ and \eager-$(r - 1)$.

Consider \lazy-$r$ and \eager-$r$.  Given a copy of~$T_\ell$ in which the root has rank~$\ell$, we can build a copy of $T_{\ell + 1}$ in which the root has rank $\ell + 1$ as follows:  First, insert an item whose key is less than that of the root, such that the new node becomes the root.  Second, do a \decreasekey\ on each item in a full child of the old root (a full grandchild of the new root), making each new key greater than that of the new root.  Third, insert an item whose key is greater than that of the new root. Finally, do a \deletemin.  Just before the \deletemin, the new root has one full child of each rank from $1$ to~$\ell$, inclusive, and two full children of rank~$0$.  In particular one of these children is the old root, which has rank $\ell$. The \deletemin\ produces a copy of $T_{\ell + 1}$.  (The \decreasekey\ operations produce hollow nodes, but no full node is a descendant of a hollow node.)  It follows by induction that one can build a copy of~$T_\ell$ for an  arbitrary value of~$\ell$ in $O(\ell^2)$ operations.  These operations followed by~$\ell^2$ repetitions of an \hinsert\ followed by a \deletemin\  form a sequence of $m=O(\ell^2)$ operations that take $\Omega(\ell^3)=\Omega(m^{3/2})$ time.

\begin{figure}[t]
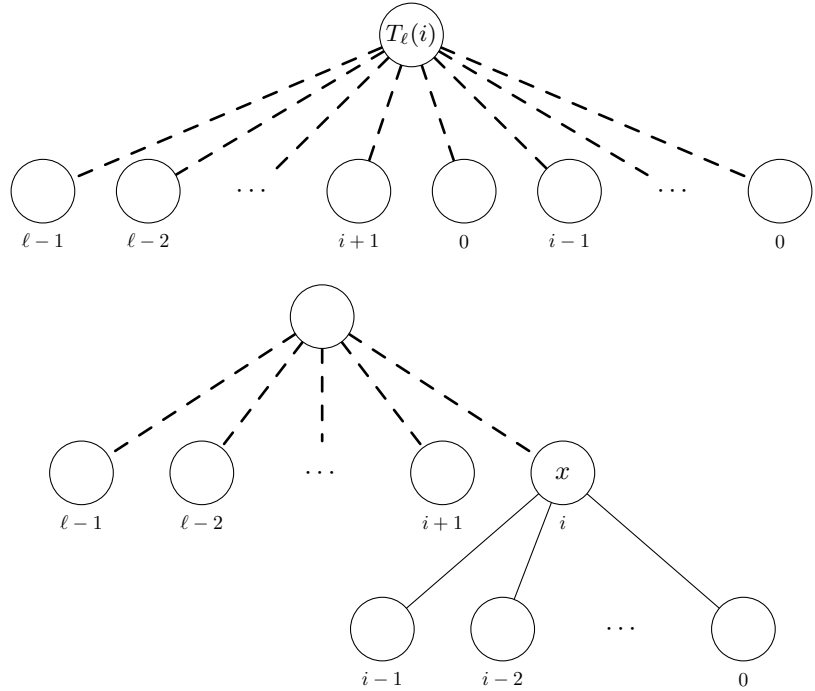

\begin{center}
\includegraphics[scale=1]{figures2.14}

\vspace{0.4cm}

\includegraphics[scale=1]{figures2.15}
\end{center}
\caption{
The construction for \lazy-$(r-1)$ and \eager-$(r-1)$. Only full nodes are shown. Solid and dashed lines denote ranked and unranked links, respectively. Ranks are shown beneath nodes. (Top) The tree $T_\ell(i)$. (Bottom) The tree obtained from $T_\ell(i)$ by inserting an item and performing a $\deletemin$ operation.
}\label{fig:eager-(r-1)}
\end{figure}



A similar but more elaborate example is bad for \lazy-$(r - 1)$ and \eager-$(r - 1)$.  Let~$T_\ell(i)$ be~$T_\ell$ with the child of rank~$i$ replaced by a child of rank~$0$.  In particular, $T_\ell(0)$ is~$T_\ell$, and $T_{\ell + 1}(\ell + 1)$ is~$T_\ell$ with the root having a second child of rank~$0$. $T_\ell(i)$ is shown at the top of Figure~\ref{fig:eager-(r-1)}.

Given a copy of $T_\ell(i)$ with $i > 0$, we can build a copy of $T_\ell(i - 1)$ as follows: First, insert an item whose key is greater than that of the root but less than that of all other items.  Now the root has three children of rank~$0$.  Second, do a \deletemin.  The just-inserted node will become the root, the other children of the old root having rank less than~$i$ will be linked by ranked links to form a tree whose root~$x$ has rank~$i$ and is a child of the new root, and the remaining children of the old root will become children of the new root.  Node~$x$ has exactly one full proper descendant of each rank from $0$ to $i - 1$, inclusive. The tree obtained after performing the \deletemin\ operation is shown at the bottom of Figure~\ref{fig:eager-(r-1)}. (In the figure we assume that the key of the child of the old root of rank $j < i$ is smaller than the key of the child of the old root of rank $j-1$ for every $1\le j  < i$. In this case~$x$ is the child of rank $i-1$ of the old root and its children after the \deletemin\ are the children of the old root of rank  $\le i-2$. But
unlike the situation shown in the figure, the descendants of~$x$ can in general be linked arbitrarily.) Finally, do a \decreasekey\ on each of the items in the full proper descendants of~$x$ in a bottom-up order (so that each \decreasekey\ is on an item in a node with only hollow descendants), making each new key greater than that of the root. The rank of each new node created this way is 1 smaller than the rank of the node it came from, except for the node that already has rank 0. The root thus gets two new children of rank 0 and one new child of each rank from 1 to $i-2$. The result is a copy of $T_\ell(i - 1)$, with some extra hollow nodes, which we ignore.  We can convert a copy of $T_\ell(0) = T_\ell$ into a copy of  $T_{\ell + 1}(\ell + 1)$ by inserting a new item with key greater than that of the root.  It follows by induction that one can build a copy of~$T_\ell$ in $m=O(\ell^3)$ operations.  These operations followed by $\ell^3$ repetitions of an \hinsert\ followed by a \deletemin\ take a total of $\Omega(\ell^4)=\Omega(m^{4/3})$ time but the desired time is $O(\ell^3\log \ell)=O(m\log m)$.

\subsection{\lazy-$k$}\label{sec:lazy-k}

Finally, we consider \lazy-$k$ for any~$k$ in the small regime. We again construct a tree for which we can repeat an expensive sequence of operations. We first give a construction for \lazy-$0$ and then show how to generalize the construction to all choices of $k$ in the small regime.

Define the tree $S_\ell$ inductively as follows. Tree $S_0$ is a single node. For $\ell> 0$,  $S_\ell$ is a tree with a full root of rank $\ell$, having one hollow child that is the root of $S_{\ell-1}$ and having full children of ranks $0,1,\ldots, \ell - 1$, with the $i$-th full child being the root of a copy of~$B_i$. The tree $S_\ell$ is shown at the top of Figure~\ref{fig:S_n}. Let~$R_\ell$ be a tree obtained by linking copies of $S_0,S_1,\ldots, S_{\ell - 1}$ to~$S_\ell$, with the root of~$S_\ell$ winning every link. The tree $R_\ell$ is shown at the bottom of Figure~\ref{fig:S_n}. We show how to build a copy of~$R_\ell$ for any~$\ell$.  Then we show how to do an expensive sequence of operations that starts with a copy of~$R_\ell$ and produces a new one.  By building one~$R_\ell$ and then doing enough repetitions of the expensive sequence of operations, we get a bad example.

\begin{figure}[t]
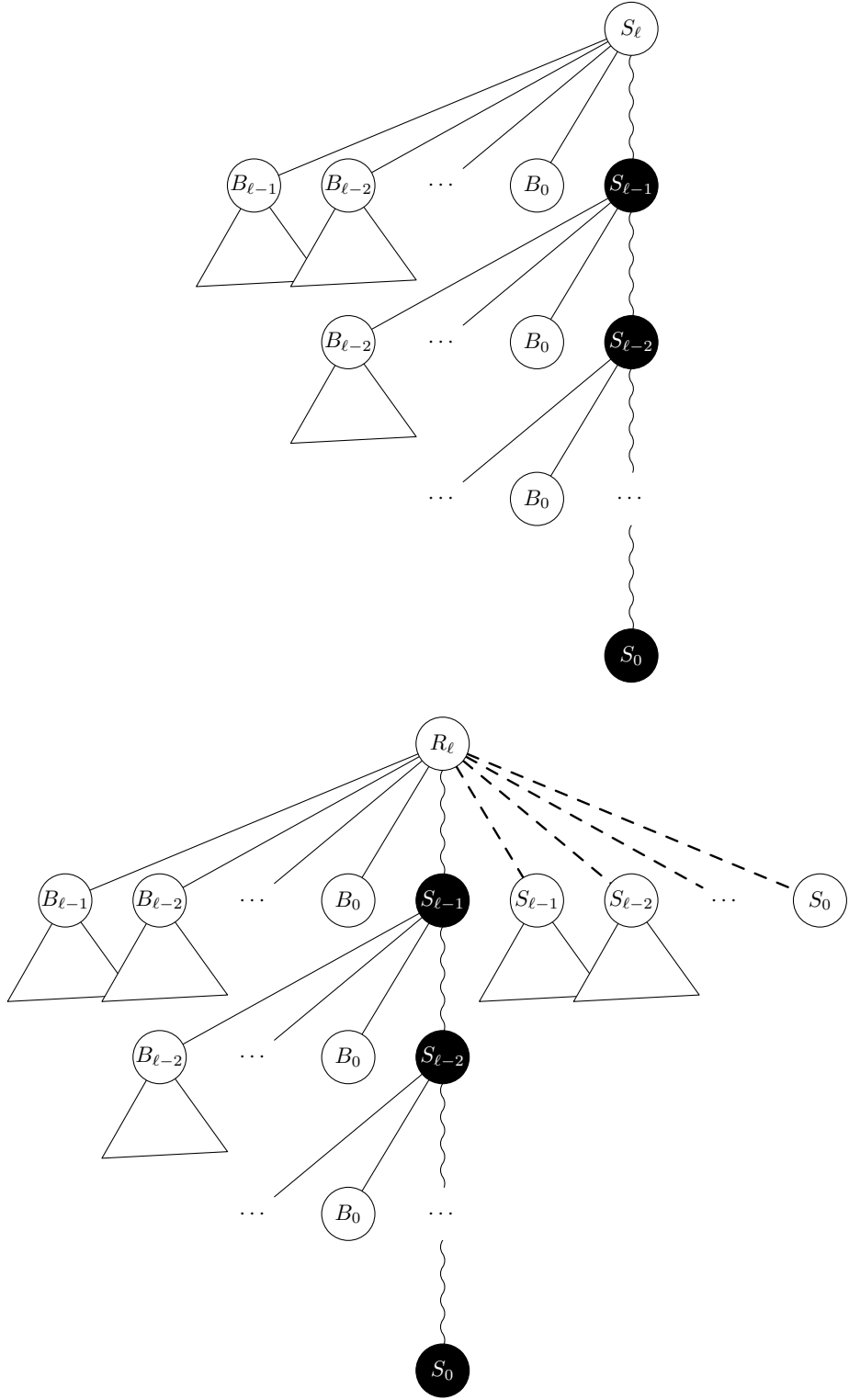

\begin{center}
\includegraphics[scale=0.91]{figures2.5}

\vspace{0.4cm}

\includegraphics[scale=0.91]{figures2.6}
\end{center}
\caption{
The trees $S_\ell$ (top) and $R_\ell$ (bottom). Every node is labeled by the type of its subtree. The triangles denote such subtrees. Black nodes are hollow. Solid and dashed lines denote ranked and unranked links, respectively. Squiggly lines denote second parents.
}\label{fig:S_n}
\end{figure}

To build a copy of~$R_\ell$ for arbitrary~$\ell$, we build a related tree~$Q_\ell$ that consists of a root whose children are the roots of copies of $S_0, S_1,\ldots, S_\ell$, with the root of~$S_\ell$ having the smallest key among the children of the root of~$Q_\ell$.  We obtain~$R_\ell$ from~$Q_\ell$ by doing a \deletemin.

We build $Q_0, Q_1,\ldots, Q_\ell$ in succession.  Tree~$Q_0$ is just a node with one full child of rank $0$, obtainable by a \makeheap\ and two \hinsert\ operations.  Given~$Q_j$, we obtain~$Q_{j + 1}$ by a variant of the construction for \eager-0. Let $x_i$ be the root of the existing copy of $S_i$ for $i = 0, \ldots, j$. In the following, all new keys are greater than the key of the root, so that the root remains the same throughout the sequence of operations. First we do \decreasekey\ operations on the roots $x_0,x_1,\ldots,x_j$ of the existing copies of $S_0, S_1,\ldots, S_j$. For $i = 0, \ldots, j$, the node $x_i$ is thus made hollow and becomes a child of a new node $y_i$ of rank 0. Note that a copy of $S_{i+1}$ can be obtained from repeated, ranked linking of $y_i$ and $2^{i+1}-1$ nodes of rank 0 where $y_i$ wins every link in which it participate. We next do enough \hinsert\ operations to provide the nodes to build $S_1,S_2,\ldots,S_{j+1}$ in this way. The total number of nodes needed is $\sum_{i=0}^j (2^{i+1}-1)$. Finally, we do two additional \hinsert\ operations, followed by a \deletemin. The two extra nodes are for a copy of $S_0$ and for a new root when the old root is deleted.

Deletion of hollow roots by \deletemin\ makes $y_i$ the only parent of $x_i$ for all $i = 0, \ldots, j$. We are left with a collection of $1+\sum_{i=0}^{j+1} 2^{i}$ roots of rank 0. We do ranked links to build the needed copies of $S_0, S_1,\ldots, S_{j+1}$ in decreasing order. Finally, we link the new root with each of the roots of the new copies of $S_i$.



Suppose we are given a copy of~$R_\ell$.  Let~$x_j$, for $j = 0, 1,\ldots, \ell$, be the root of the copy of $S_\ell$.  In particular, $x_\ell$ is the root of~$R_\ell$. We can do an expensive sequence of operations that produces a new copy of~$R_\ell$ as follows: Do \decreasekey\ operations on~$x_j$ for $j = 0, 1,\ldots, \ell - 1$, giving each~$x_j$ a second parent $y_j$.  Make all the new keys larger than that of $x_\ell$ and smaller than those of all children of $x_\ell$; among them, make the key of $y_{\ell - 1}$ the smallest.  Next, insert a new item with key greater than that of $y_{\ell - 1}$; let $z$ be the new node holding the new item. Figure~\ref{fig:lazy-0_construction} shows the resulting situation. Next, do a \deletemin.  This makes $y_j$ the only parent of $x_j$ for $j = 0,\ldots, \ell - 1$. Once the hollow roots are deleted, the remaining roots are $z$, the $y_j$, and the roots of $\ell - i$ copies of $B_i$ for $i = 0, 1,\ldots , \ell - 1$. Finish the \deletemin\ by doing ranked links of each $y_j$ with the roots of copies of $B_i$ for $i = 0,1,\ldots, j$, forming new copies of  $S_0,S_1,\ldots, S_\ell$ ($z$ is the root of a copy of $S_0$; $y_j$ is the root of a copy of $S_{j + 1}$), and link the roots of these copies by unranked links.  The result is a new copy of~$R_\ell$.  The sequence of operations consists of one \hinsert, $\ell$ \decreasekey\ operations, and one \deletemin\ and takes  $O(\ell^2)$ time.

\begin{figure}[t]
\begin{center}
\includegraphics[scale=0.91]{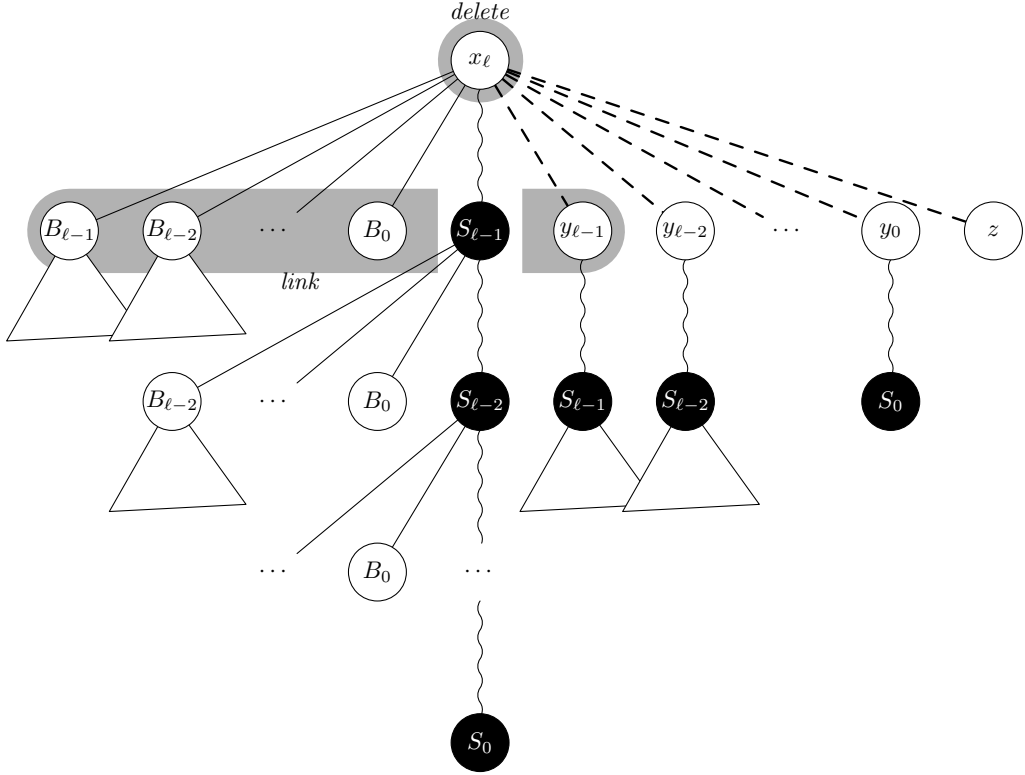}
\end{center}
\caption{
The tree obtained from $R_\ell$ by performing $\ell$ \decreasekey\ operations and one \hinsert. Every node is  labeled by its name or the type of its subtree. The triangles denote such subtrees. Black nodes are hollow. Solid and dashed lines denote ranked and unranked links, respectively. Squiggly lines denote second parents. Edges connecting $x_\ell$ to the children of $y_i$ for $i = 0, \ldots, \ell-1$  have been omitted.
}\label{fig:lazy-0_construction}
\end{figure}

The number of nodes in~$R_\ell$ is $O(2^\ell)$, as is the number of operations needed to build it and the time these operations take. Having built $R_\ell$,  if we then do $2^\ell$ repetitions of the expensive sequence of operations described above, the total number of operations is $m = O(\ell2^\ell)$.  The operations take  $\Theta(\ell^22^\ell) = \Theta(m\ell)$ time, whereas the desired time is $O(m)$.

An extension of the same construction shows the inefficiency of \lazy-$k$ for~$k$ in the small regime: instead of doing one \decreasekey\ on each appropriate item, we do enough to reduce to $0$ the rank of the full node holding the item.  Suppose $k = r - f(r)$, where $f(r)$ is a positive non-decreasing function tending to infinity.  Then the number of \decreasekey\ operations needed to reduce the rank of the node holding an item to $0$, given that the rank of the initial node holding the item is~$k$, is at most $k/f(\sqrt{k}) + \sqrt{k}$. It follows that the extended construction does at most $\ell^2/f(\sqrt{\ell}) + \ell^{3/2}$ \decreasekey\ operations per round, and the amortized time per decrease-key is $\Omega(f(\sqrt{\ell}))= \omega(1)$, assuming that the amortized time per \hdelete\ is $O(\ell)$ and that of \makeheap\ and \hinsert\ is $O(1)$.

\end{document}